\documentclass[aps,prl,reprint,superscriptaddress,nofootinbib,floatfix,longbibliography]{revtex4-2}
\usepackage{amsmath,amssymb,mathtools,bm,amsthm}
\usepackage{graphicx,tikz}
\usetikzlibrary{arrows.meta,positioning,calc}
\usepackage{microtype}
\usepackage{enumitem}
\usepackage{braket}
\usepackage{xcolor}
\usepackage{hyperref}
\hypersetup{colorlinks=true,citecolor=blue!45!black,linkcolor=blue!45!black,urlcolor=blue!45!black}
\newcommand{\F}{\mathbb F}
\newcommand{\E}{\mathbb E}
\newcommand{\Cliff}{\mathrm{Cliff}}
\newcommand{\Tr}{\operatorname{Tr}}
\newcommand{\Span}{\operatorname{span}}
\newcommand{\rank}{\operatorname{rank}}
\newcommand{\RM}{\mathrm{RM}}
\newcommand{\CTP}{\mathcal C}
\newcommand{\supp}{\operatorname{supp}}

\newcommand{\CPW}{w_{\mathrm{CP}}}
\newcommand{\snull}{\nu}
\newcommand{\cL}{\mathcal L}
\newcommand{\Dtr}{D_{\rm tr}}
\newcommand{\prlsection}[1]{{\em #1.---}}

\newcommand{\SMLong}{Supplementary Material}

\newcommand{\SM}{SM}

\newtheorem{theorem}{Theorem}
\newtheorem{lemma}[theorem]{Lemma}
\newtheorem{proposition}[theorem]{Proposition}
\newtheorem{corollary}[theorem]{Corollary}
\begin{document}

\title{Efficient Learning of Clifford-Scrambled Product States}
\author{Tobias Haug}
\email{tobias.haug@u.nus.edu}
\affiliation{Quantum Research Center, Technology Innovation Institute, Abu Dhabi, UAE}

\begin{abstract}
Clifford circuits acting on product magic states provide a compact ansatz exhibiting extensive magic, volume-law entanglement, and even classically hard sampling under standard complexity assumptions.
Here, we efficiently recover its hidden product subsystems from two-copy Bell sampling. 
Quadratic relations between Bell samples determine the irreducible blocks after removing Pauli stabilizers, and binary linear algebra constructs a Clifford disentangler.
For logarithmic-size blocks, approximate learning of the full state is efficient whenever the state remains inverse-polynomially separated from acquiring additional Pauli stabilizers. 
We efficiently learn $n$-qubit states prepared by random $T$-doped Clifford circuits with $T$-gate density below one via $O(n^2)$ Bell samples, and disentangle hidden product blocks in Clifford-augmented matrix product states.
For unitary learning, we exactly learn $T$-depth-one circuits using $O(n^2)$ queries.
Finally, we rule out pseudorandom states and unitaries with a product bipartition hidden by Clifford circuits.
\end{abstract}

\maketitle

\let\oldaddcontentsline\addcontentsline%
\renewcommand{\addcontentsline}[3]{}%

Learning a quantum state from measurements is a central task for characterizing quantum devices. Full tomography requires resources that grow exponentially with the number of qubits, which makes assumptions about the state essential~\cite{HaahEtAl2017}. Even a short preparation circuit, although sufficient for sample efficiency, does not generally make reconstruction computationally efficient~\cite{ZhaoEtAl2024Learning}. 
However, structural assumptions can make learning efficient~\cite{HuangEtAl2024Shallow,ArunachalamEtAl2023Phase}:
for example, matrix-product states (MPSs) admit efficient tomography when their bond dimensions and reconstruction conditioning are favorable~\cite{CramerEtAl2010}. Stabilizer states provide a different route: they can have volume-law entanglement, but their Pauli symmetries can be learned efficiently~\cite{montanaro2017learning}. Algorithms for states prepared with few non-Clifford gates extend this approach beyond stabilizers, using either collective or single-copy measurements~\cite{LeoneEtAl2024,ChiaLaiLin2024,GrewalEtAl2025}.

Recent results further allow agnostic learning of states with low stabilizer nullity~\cite{ChenGongYeZhang2025} and tomography controlled by stabilizer extent~\cite{ArunachalamDutt2026Extent}. %
They leave open a different regime: a state may contain many independent magic resources, each small, while having maximal stabilizer nullity. Which structure permits efficient learning when both entanglement and nonstabilizerness, or magic, are extensive?

Clifford circuits acting on independent magic states give a concrete setting for this question. We consider pure states prepared by applying a global Clifford circuit to independent blocks of qubits. Each block has a simple local description, while the Clifford can spread its degrees of freedom across the entire system. %

Even for single-qubit inputs, these states can have extensive magic and volume-law entanglement. 
In particular, the family $C|T\rangle^{\otimes n}$, with $C$ Clifford and $|T\rangle=(|0\rangle+e^{i\pi/4}|1\rangle)/\sqrt2$, contains output distributions that are classically hard to sample to multiplicative error if the polynomial hierarchy does not collapse; additive-error hardness requires additional average-case conjectures~\cite{YoganathanEtAl2019}. 
Such states also form the basis of universal quantum computing: combining them with measurements and Clifford feed-forward enables universal quantum algorithms.

Clifford-augmented MPSs~\cite{QianHuangQin2024,LamiHaugDeNardis2025} and stabilizer tensor networks~\cite{masot2024stabilizer} combine a Clifford frame with a compact description of the remaining quantum state. Single-qubit products give the bond-dimension-one case, while products of entangled blocks allow larger residual bond dimensions. Such representations improve simulations of ground states, time evolution~\cite{QianHuangQin2024,MelloEtAl2025TDVP,QianHuangQin2025TDVP}, and $T$-doped Clifford circuits~\cite{FuxEtAl2025}, however existing disentangling methods use heuristic optimization and circuit-aware algebraic methods on classical  descriptions%
~\cite{FuxEtAl2025,LiuClark2026,MasotLlimaEtAl2026}. 

When a target Clifford and its input product states are supplied, efficient protocols can certify Clifford-augmented products~\cite{AbdulSaterEtAl2025}. 
Learning instead requires inferring both the Clifford frame and the input states. 
Without Clifford scrambling, state hidden-subgroup algorithms can recover unknown partitions of physical qubits~\cite{BoulandEtAl2025}, including an implementation with two-copy Bell measurements and an extension to product unitaries through Choi states~\cite{HinscheEtAl2026}. Global Cliffords can scramble the physical product cuts, so recovering the independent subsystems requires finding the Clifford frame as well as the partition.
While two-copy measurements already give product tests~\cite{HarrowMontanaro2013} and Bell-sampling protocols for Pauli information and nonlinear state properties~\cite{montanaro2017learning,hangleiter2024bell}, it has remained open how to recover the hidden product structure of Clifford-scrambled states. %

Here, we show that two-copy Bell sampling can efficiently recover the hidden product structure of Clifford-scrambled states. 
In particular, quadratic equations satisfied by Bell samples determine every hidden product cut once Pauli stabilizers are removed. 
These equations yield a unique symplectic decomposition into irreducible blocks, recovered in time polynomial in system size and number of samples. 
For logarithmic block size and an inverse-polynomial quadratic support gap, this gives efficient full tomography of the state. 
Even when the quadratic gap becomes arbitrarily small, we can  approximately learn the state as long as it stays inverse-polynomially away from acquiring additional Pauli stabilizers. 
Applied to process learning of unknown $T$-depth-one unitaries, we recover an exact circuit implementation with probability at least $1-\delta$ using $O(n^2+\log\delta^{-1})$ queries. %
We also efficiently identify Clifford-scrambled products of $T$-states and the output states of random $T$-doped Clifford circuits at $T$-gate density below one. More generally, our algorithm can construct Clifford-augmented MPS representations when the residual state is a product of logarithmic-size blocks.
Finally, we rule out pseudorandom states and unitaries with a product bipartition hidden by Clifford circuits. 

\begin{figure}[t]
\centering
\resizebox{0.99\columnwidth}{!}{
\begin{tikzpicture}[x=1cm,y=1cm,>=Latex,font=\small,
 box/.style={draw,rounded corners=2pt,inner sep=4pt},
 wire/.style={line width=.65pt},
 lab/.style={font=\footnotesize,align=center}]
 \node[anchor=west,font=\bfseries\footnotesize] at (0,2.2) {(a) Clifford-scrambled product state};
 \node[box,fill=orange!12] (phi1) at (.65,1.62) {$\phi_1$};
 \node[box,fill=orange!12] (phi2) at (.65,1.03) {$\phi_2$};
 \node[box,fill=orange!12] (phi3) at (.65,.44) {$\phi_3$};
 \node[box,fill=blue!10,minimum width=1.1cm,minimum height=1.8cm] (cliff) at (2.2,1.03) {$C$};
 \foreach \p in {phi1,phi2,phi3}
   \draw[wire] (\p.east)--(cliff.west |- \p.east);
 \node[box,fill=blue!5,minimum width=1.25cm,minimum height=.6cm] (psi) at (4.13,1.03) {$|\psi\rangle$};
 \draw[->,wire] (cliff.east)--(psi.west);
 \node[anchor=west,font=\bfseries\footnotesize] at (0,-.18) {(b) Learning from Bell samples};
 \node[box,fill=blue!5,lab,minimum width=1.2cm] (bell) at (.73,-.87) {two-copy\\Bell samples};
 \node[box,fill=gray!9,lab,minimum width=1.75cm,right=.35cm of bell] (aff) {remove Pauli\\stabilizers};
 \node[box,fill=green!8,lab,minimum width=1.9cm,right=.35cm of aff] (quad) {quadratic equations\\give product cuts};
 \node[box,fill=orange!12,lab] (product) at (2.86,-1.80) {$D^\dagger|\psi\rangle=\bigotimes_\alpha|\phi_\alpha\rangle$};
 \draw[->,wire] (bell.east)--(aff.west);
 \draw[->,wire] (aff.east)--(quad.west);
 \draw[->,wire] (quad.south)|-(product.east);
\end{tikzpicture}
}
\caption{Learning product subsystems hidden by a Clifford circuit. 
(a) The unknown Clifford $C$ spreads the input blocks across the system and can generate volume-law entanglement. 
(b) We measure pairs of copies in the Bell basis. Affine equations identify Pauli stabilizers, and quadratic equations on the residual samples determine the hidden subsystems. Applying the learned Clifford decoder $D^\dagger$ to fresh copies disentangles the subsystems, followed by tomography of each block. %
}
\label{fig:learner}
\end{figure}
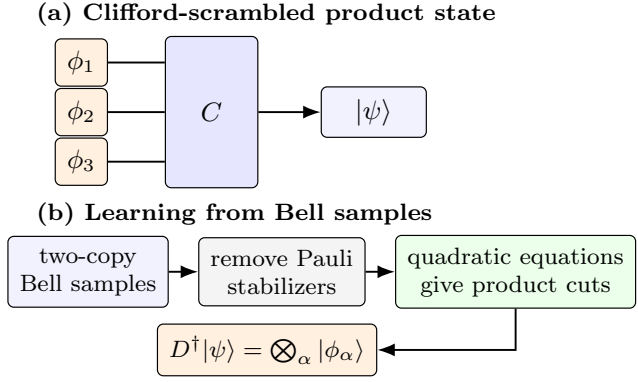

\prlsection{Clifford-scrambled states}
Label Hermitian $n$-qubit Paulis by $v=(x,z)\in V=\F_2^{2n}$, with
$W_v=\bigotimes_{j=1}^n i^{x_jz_j}X_j^{x_j}Z_j^{z_j}$.
The labels $(0,0),(1,0),(0,1),(1,1)$ on one qubit represent $I,X,Z,Y$, respectively. Label arithmetic is over $\F_2$; the phase in $W_v$ uses ordinary integer exponents. Pauli commutation $W_uW_v=(-1)^{[u,v]}W_vW_u$ is encoded by the symplectic form $[u,v]=u^TJv$ with $J=\bigl(\begin{smallmatrix}0&I\\I&0\end{smallmatrix}\bigr)$.
A Clifford unitary $C$ maps Pauli operators  to Pauli operators under
conjugation
\begin{equation}
 C W_a C^\dagger=\pm W_{F_Ca},
 \qquad F_C^TJF_C=J,
\end{equation}
where the binary symplectic matrix $F_C$ preserves Pauli commutation relations.
The $n$-qubit Clifford group, denoted by $\Cliff_n$,
is generated by Hadamard, phase, and CNOT gates, up to global phases.

A Pauli stabilizer of $|\psi\rangle$ is a signed Pauli operator
$P=\pm W_a$ satisfying $P|\psi\rangle=|\psi\rangle$. Ignoring these
signs, the stabilizer labels form the binary subspace
\begin{equation}
 \mathsf S_\psi
 =\{a:|\langle\psi|W_a|\psi\rangle|=1\}.
\end{equation}
All stabilizers commute. A pure stabilizer state has $n$ independent
Pauli stabilizers, which uniquely specify the state together with
their signs. A general pure state may have fewer; we call it
stabilizer-free if it has no nonidentity Pauli stabilizer.

We define the class of $n$-qubit Clifford-scrambled product states with blocks of at most $k$ qubits as
\begin{equation}
 \CTP_{n,k}=\left\{C\bigotimes_{\alpha=1}^{r}|\phi_\alpha\rangle_{A_\alpha}:
 C\in\Cliff_n,\ \max_\alpha|A_\alpha|\le k\right\},
 \label{eq:class}
\end{equation}
where $\{A_\alpha\}_{\alpha=1}^r$ partitions the qubits into $r$ blocks and each $|\phi_\alpha\rangle$ is a pure block state. The Clifford $C$, partition, and block states are all unknown.

To quantify the smallest blocks needed to describe a state in any Clifford frame, we define the \emph{Clifford-product width},
\begin{equation}
 \CPW(\psi)=\min\{k:\psi\in\CTP_{n,k}\}.
 \label{eq:cp-width}
\end{equation}
It is the minimum achievable largest block size over all Clifford transformations and partitions. %

\prlsection{Quadratic equations encode product cuts}
To recover a Clifford-hidden subsystem, we first express product structure as a constraint on Bell samples. Measuring all corresponding qubit pairs of two identical copies in the Bell basis produces one Bell sample: the complete $2n$-bit label drawn from $p_\psi(v)=2^{-n}|\langle\psi^*|W_v|\psi\rangle|^2$~\cite{montanaro2017learning,hangleiter2024bell}. 

Bell measurements uncover the product structure of a given state. In particular, swapping the two copies on a subset $A$ has Bell eigenvalue $(-1)^{q_A(v)}$, where $q_A(v)=\sum_{i\in A}x_i z_i$. Its expectation gives the reduced-state purity, so
\begin{equation}
 \Pr_{v\sim p_\psi}[q_A(v)=1]=\frac{1-\Tr\rho_A^2}{2}.
 \label{eq:purity}
\end{equation}
Thus a pure state is a tensor product across $A|\bar A$ precisely when every label with nonzero probability satisfies $q_A=0$, where $\bar A$ is the complement of $A$. The full swap supplies one relation common to all pure states, $q=q_{[n]}=0$.

A Clifford relabels Bell samples by $g_C(v)=F_Cv+t_C$. This map is \emph{affine}: it combines a linear change of binary coordinates with a fixed shift $t_C$. Its linear part preserves Pauli commutation, and the full map preserves $q$. This relabeling preserves quadratic degree even for subsystems spread across all physical qubits.

We first identify \emph{affine equations}: parity checks $h(v)=c+[a,v]$ on the bits of each Bell sample, with prescribed parity $c\in\F_2$. They detect Pauli stabilizers because
\begin{equation}
 \E(-1)^{h(v)}=(-1)^{c+q(a)}\langle W_a\rangle^2.
 \label{eq:affine}
\end{equation}
The left side equals one precisely when $h(v)=0$ on every label in the Bell support. For $a\ne0$, this requires $|\langle W_a\rangle|=1$ and $c=q(a)$, so $\pm W_a$ stabilizes the state. Bell samples determine these labels; a separate measurement determines the stabilizer signs.

Established stabilizer-learning methods can recover the labels~\cite{montanaro2017learning,GrewalEtAl2025,HinscheEtAl2026} and a Clifford separating $n-\snull$ fixed qubits, where $\snull=n-\dim\mathsf S_\psi$ is the stabilizer nullity~\cite{beverland2020lower}. Relabeling the samples and deleting their fixed-qubit coordinates leaves Bell samples from the residual $\snull$-qubit state. This removal is essential: multiplying an exact affine equation by a linear function also gives a vanishing quadratic, but supplies no new independent subsystem (see \SMLong{} (\SM{})~\ref{sm:bell}). %

For the remaining stabilizer-free state, let $K_\psi$ contain all Boolean functions of degree at most two that vanish on the complete Bell support. 
To identify the subsystem associated with an equation $f\in K_\psi$,
we first extract its quadratic coefficients. The combination
$f(u+v)+f(u)+f(v)+f(0)$ cancels the constant and linear terms,
leaving a bilinear form represented by the matrix $B_f$:
\begin{equation}
 u^T B_f v=f(u+v)+f(u)+f(v)+f(0),\quad E_f=JB_f.
 \label{eq:polar}
\end{equation}
This operation is called polarization. The map $E_f$ is
self-adjoint with respect to the symplectic form,
$[E_fu,v]=[u,E_fv]$. 
We show below that, for exact equations of a stabilizer-free state, it is also a projector whose image identifies the Pauli-label space of an independent
subsystem.
For a physical block $A$, $E_{q_A}$ projects onto its Pauli phase space. A Clifford conjugates this projector by $F_C$. 
The key insight is that, after removing Pauli stabilizers,
every quadratic equation satisfied by all possible Bell samples
expresses invariance under a partial swap in some Clifford frame,
and therefore identifies a product cut.
\begin{theorem}[Quadratic classification]\label{thm:structure}
For a stabilizer-free pure state, every $E_f$, $f\in K_\psi$, is idempotent. These maps form a commutative Boolean algebra whose minimal nonzero projectors define unique irreducible symplectic subspaces, up to ordering. In that frame, $K_\psi$ is exactly the span of the block-swap equations.
\end{theorem}

We call these blocks \emph{irreducible}: no Clifford acting within a block can split its state into a product on smaller subsystems.

The proof uses the constraint that an exact Bell equation imposes on Pauli correlations. The Bell-diagonal involution $Z_f=\sum_v(-1)^{f(v)}|B_v\rangle\langle B_v|$ satisfies $Z_f|\psi\rangle^{\otimes2}=|\psi\rangle^{\otimes2}$: every Bell component with nonzero amplitude has $f(v)=0$ and hence eigenvalue $+1$. Here, $|B_v\rangle=(I\otimes W_v)|\Phi\rangle$ and
$|\Phi\rangle=2^{-n/2}\sum_j|j,j\rangle$. Conjugating $W_a\otimes I$ by $Z_f$ gives, up to phase, $W_{(I+E_f)a}\otimes W_{E_fa}^*$. Taking expectations yields
\begin{equation}
 R_\psi(a)=R_\psi(E_fa)R_\psi((I+E_f)a),
 \quad R_\psi(a)=|\langle W_a\rangle|.
 \label{eq:rigidity}
\end{equation}
Suppose $E_f^2a\ne E_fa$ for a label with nonzero expectation, and choose such a label maximizing $R_\psi(a)$. Both labels on the right-hand side of Eq.~\eqref{eq:rigidity} are nonzero, and at least one also violates idempotence. Since a stabilizer-free state has $R_\psi(b)<1$ for every nonzero $b$, both have larger expectation than $a$, a contradiction. The labels with nonzero expectation span $V$, which gives $E_f^2=E_f$. Applying this result to $E_f+E_g$ also gives $E_fE_g=E_gE_f$.

An idempotent separates $V$ into symplectically orthogonal image and kernel. In the corresponding Clifford frame, $f$ is a partial-swap equation plus an affine term. Bounding the resulting expectation of the swap combined with Pauli operators by $\Tr\rho_A^2$ forces the purity to equal one. Thus $f$ identifies a product cut. Intersecting the compatible cuts gives the irreducible blocks and proves uniqueness. The quadratic kernel therefore determines all possible hidden cuts. We give the full proof in the \SM{}~\ref{sm:classification}.

\prlsection{Learning the hidden subsystems}
Theorem~\ref{thm:structure} shows that the irreducible decomposition attains the Clifford-product width, including the fixed qubits as single-qubit blocks. Recovering this decomposition therefore identifies the smallest possible largest block in any Clifford frame.

To recover the equations from finitely many copies, we must resolve how often a false equation is violated. Write $\RM(j,2n)$ for Boolean functions of degree at most $j$ on $\F_2^{2n}$ and $K_\psi^{(j)}$ for the exact kernel. Define the support gaps
\begin{equation}
\gamma_j(\psi)=\min_{f\in\RM(j,2n)\setminus K_\psi^{(j)}}\Pr_{p_\psi}[f(v)=1],
 \qquad j=1,2.
 \label{eq:gaps}
\end{equation}
Each gap is the smallest probability of observing a violation of
an equation $f(v)=0$ that does not hold exactly on the Bell
distribution. A small gap therefore requires many samples to
distinguish an exact equation from one that is only rarely violated.
The affine gap $\gamma_1$ controls stabilizer recovery, while
the quadratic gap $\gamma_2$ controls recovery of the hidden
product structure. 
Since affine polynomials are included among
quadratic polynomials, $\gamma_1\ge\gamma_2$, so a lower bound on
$\gamma_2$ also guarantees efficient stabilizer recovery.

We evaluate all monomials of degree at most two on the Bell samples and compute their binary nullspace. This retains exactly the equations consistent with every sample. There are $D_2(n)=1+2n+\binom{2n}{2}=O(n^2)$ monomials, so a union bound over $2^{D_2(n)}$ candidate equations controls the probability of retaining a false one.

\begin{theorem}[Exact product structure recovery]\label{thm:exact}
For $\psi\in\CTP_{n,k}$ with $\gamma_2(\psi)\ge\gamma$, 
a Clifford decoder and the irreducible blocks of the stabilizer-free residual are recovered with probability at least $1-\delta$ using
\begin{equation}
 M=O\!\left[\gamma^{-1}(n^2+\log\delta^{-1})\right]
 \label{eq:exact-cost}
\end{equation}
Bell samples and classical time polynomial in $n$ and $M$, with the largest recovered block size $\CPW(\psi)$.
\end{theorem}
The algorithm first removes the affine kernel and then polarizes the remaining quadratic equations.
The joint eigenspaces of these commuting projectors give the hidden subsystems. Choosing a symplectic basis in each subsystem specifies the decoder's Clifford tableau, avoiding a search over Clifford circuits. The same Bell samples suffice for both stages because every equation after projection lifts to a quadratic equation on the original labels. The residual symplectic subspaces are unique; local Clifford bases, Pauli phases, and block order remain free choices.

Applying the decoder to fresh copies allows us to perform pure-state tomography on all blocks in parallel to global trace-distance error $\varepsilon$ using copies  %
\begin{equation}
O\!\left(r2^k\varepsilon^{-2}
\log(r/\delta)\right).
\label{eq:tomo}
\end{equation}
We obtain this bound by allocating the global infidelity budget among the blocks and applying efficient pure-state tomography with standard confidence amplification~\cite[Appendix~C]{HaahKothariODonnellTang2023}. The width bounds each local Hilbert-space dimension by $2^{\CPW(\psi)}$. Thus logarithmic width and an inverse-polynomial gap give a polynomial-time learner whose output is a Clifford circuit and pure block states in the promised class.

For example, every state $C|T\rangle^{\otimes n}$ has width one and maximal stabilizer nullity $n$. Nevertheless, $\gamma_2=1/16$ for $n\ge2$ (\SM{}~\ref{sm:qubitgaps}), so $O(n^2)$ Bell samples recover its structure at constant confidence for any Clifford $C$. 
This example extends to output states of random Clifford circuits doped with $T$ gates~\cite{haferkamp2023efficient,leone2021quantum}. Their product structure below the disentangling threshold~\cite{FuxEtAl2025,LiuClark2026} enables exact state identification. We make the ensemble and success probability explicit in the following corollary.

\begin{corollary}[Learning random $T$-doped Clifford states]
\label{cor:random-t-learning}
Starting from $|0\rangle^{\otimes n}$, apply $t<n$ independent uniform global Cliffords, each followed by a single $T$ gate. From copies of the output, our algorithm recovers a Clifford $\widehat D$ such that
$|\psi\rangle=\widehat D(|T\rangle^{\otimes t}\otimes|0\rangle^{\otimes(n-t)})$
up to global phase, with probability at least $1-2^{t-n}-\delta$ over the circuit and measurements. For $0<\delta<1$, it uses
\begin{equation}
 O\!\left(n^2+\log\delta^{-1}\right)
 \label{eq:random-t-learning-cost}
\end{equation}
copies and time polynomial in $n,\log\delta^{-1}$.
\end{corollary}
The circuit contribution $2^{t-n}$ to failure is exponentially small when $t\le(1-\eta)n$ for fixed $0<\eta<1$. After frame recovery, the decoded magic qubits lie in a finite single-qubit Clifford orbit, so constant-accuracy Pauli measurements identify them exactly (\SM{}~\ref{sm:random-t-learning}). 

More generally,  \SM{}~\ref{sm:qubitgaps} also considers products of $n\ge2$ nonstabilizer qubits: we find a gap $\gamma_2=\eta_{(1)}\eta_{(2)}$, where $\eta_i$ is the smallest positive Bell probability of qubit $i$ and $\eta_{(1)},\eta_{(2)}$ are the two smallest of these local minima. For independent Haar qubits, the typical gap scales as $n^{-2}$.

\prlsection{Approximation when the quadratic gap closes}
Exact reconstruction requires resolving rare violations that distinguish an exact product cut from a weakly entangled one. However, for approximate tomography we do not need to resolve entanglement below the target error. We make this statement precise by extending Theorem~\ref{thm:structure} to equations inferred from a finite sample. The resulting algorithm removes dependence on $\gamma_2$, while still depending on the affine gap (see \SM{}~\ref{sm:exact-learning})
\begin{equation}
 \gamma_1(\psi)=\frac{1-\max_{a\notin\mathsf S_\psi}R_\psi(a)^2}{2}.
 \label{eq:affinegap}
\end{equation}
This gap quantifies separation from a state with an additional Pauli stabilizer. It can stay finite even when the quadratic gap becomes arbitrarily small.

\begin{theorem}[Approximate learning]\label{thm:approx}
Let $\psi\in\CTP_{n,k}$ and $\gamma_1(\psi)\ge \gamma_\star>0$. A Clifford and pure block states of size at most $k$ can be learned to trace distance $\varepsilon$, with failure probability at most $\delta$, in time and copies polynomial in $n,2^k,\gamma_\star^{-1},\varepsilon^{-1},\log\delta^{-1}$, independently of $\gamma_2$.
\end{theorem}
After recovering and removing the exact Pauli stabilizers, write $\tau_f(\psi):=\Pr_{v\sim p_\psi}[f(v)=1]$ for the violation probability of a candidate quadratic equation and choose a tolerance $\tau>0$. From $M=O[(n^2+\log\delta^{-1})/\tau]$ Bell samples, we obtain, with high probability, only equations $f(v)=0$ with $\tau_f(\psi)\le\tau$. Each such equation implies that the Pauli correlations satisfy Eq.~\eqref{eq:rigidity} within $4\sqrt{\tau}$. %
If the associated map fails to be idempotent, the block-size promise supplies a violating Pauli correlation of magnitude at least $\sqrt{2\gamma_\star/2^k}$. Comparing this lower bound with the upper bound from the approximate correlation identity gives
\begin{equation}
 \tau<\frac{\gamma_\star^3}{8\cdot2^k}\quad\Longrightarrow\quad
 E_f^2=E_f.
 \label{eq:robust-idem}
\end{equation}
Thus sufficiently accurate empirical equations still give exact, commuting projectors. Each joint eigenspace lies within one true block because the true block equations vanish on every sample. Each empirical basis equation identifies a cut with purity at least $1-2\tau$. We use strong subadditivity to bound the sum of the block entropies and obtain
\begin{equation}
 \Dtr(\psi,\text{recovered product set})^2
 \le n^2[h_2(2\tau)+2n\tau],
 \label{eq:approx-error}
\end{equation}
where $\Dtr$ is trace distance and $h_2$ is binary entropy. For logarithmic $k$ and inverse-polynomial $\gamma_\star$, an inverse-polynomial $\tau$ controls both this error and Eq.~\eqref{eq:robust-idem}. We then perform tomography on the slightly mixed block marginals and extract pure block states, giving the stated global approximation (\SM{}~\ref{sm:approx}).

\prlsection{Learning circuits of T-depth one}
We next learn an unknown unitary process as an explicit circuit. 
We consider unitaries with magic depth one, where all the magic gates are concentrated in a single layer~\cite{zhang2025classical}
\begin{equation}
 U=C_{\rm out}(T^{\otimes t}\otimes I)C_{\rm in},
 \qquad 0\le t\le n,
 \label{eq:t-depth-one}
\end{equation}
where both Clifford circuits and $t$ are unknown. 
Earlier methods achieve polynomial-time learning for $O(\log n)$ $T$ gates~\cite{LaiCheng2022,leone2024learning}. Here, we recover a coherent implementation correct on arbitrary inputs for any $t\le n$. This regime also differs from bounded-Clifford-extent tomography~\cite{DuttEtAl2026}.
For process learning~\cite{MansurogluEtAl2024,WadhwaEtAl2025}, we start with the normalized Choi state
\begin{equation}
 |J(U)\rangle=2^{-n/2}\sum_{x\in\{0,1\}^n}
 |x\rangle_R\otimes U|x\rangle_A
 \label{eq:choi-definition}
\end{equation}
with Clifford-product width $w_J(U):=\CPW(J(U))$. %
\begin{corollary}[Learning $T$-depth-one unitaries]
\label{cor:t-depth-one}
Given query access to an $n$-qubit unitary $U$ satisfying Eq.~\eqref{eq:t-depth-one}, a quantum algorithm returns Clifford circuits $\widehat C_{\rm in},\widehat C_{\rm out}$ and $t$ such that
\begin{equation}
 \widehat C_{\rm out}(T^{\otimes t}\otimes I)\widehat C_{\rm in}
 =e^{i\varphi}U
 \label{eq:t-depth-one-output}
\end{equation}
with probability at least $1-\delta$. The algorithm uses $O(n^2+\log\delta^{-1})$ queries and polynomial time.
\end{corollary}
The identity $\mathrm{CNOT}_{R\to A}|J(T)\rangle=|T\rangle_R|0\rangle_A$ gives $w_J(U)=1$ and stabilizer nullity $t$. These Choi states have $\gamma_2\ge1/16$, so Theorem~\ref{thm:exact} recovers their product structure with the stated query cost. The decoded magic qubits belong to the finite single-qubit Clifford orbit of $|T\rangle$, allowing exact identification with high probability. We then find output-register representatives of their Pauli rotation axes by binary linear algebra. Removing these rotations leaves a stabilizer Choi state, which determines the remaining Clifford circuit (\SM{}~\ref{sm:t-depth-one}). %

For more general circuits of magic depth one with logarithmic-size magic blocks, the same reduction gives efficient Choi-state learning under our gap assumptions (\SM{}~\ref{sm:choi}). %

\prlsection{A hidden cut forbids pseudorandomness}
We now use our method to distinguish hidden product structure from Haar randomness. Pseudorandom states are efficiently generated ensembles that no polynomial-time observer with polynomially many copies can distinguish from Haar-random states~\cite{JiLiuSong2018}. We test for hidden structure by accepting whenever the empirical quadratic kernel has dimension greater than one. Ordinary binary elimination implements this rank test in $O(n^6)$ bit operations and $O(n^4)$ bits of memory for $O(n^2)$ Bell samples; no Clifford reconstruction is needed (\SM{}~\ref{sm:pseudorandomness}). Applying the test to Choi copies also distinguishes unitary ensembles.

\begin{corollary}\label{cor:prs}
An efficiently generated $n$-qubit state ensemble placing nonnegligible weight on $\CPW(\psi)<n$ is not a pseudorandom-state ensemble. Likewise, an efficiently generated $n$-qubit unitary ensemble placing nonnegligible weight on $w_J(U)<2n$ is not a pseudorandom-unitary ensemble. The distinguishers use $O(n^2)$ copies or queries and polynomial processing.
\end{corollary}
Here, negligible means decaying faster than any inverse polynomial.
A hidden cut supplies an exact equation independent of the full swap $q$, which survives any number of samples. For Haar-random states, we prove $\gamma_2=1/4-o(1)$, so only $q$ survives $O(n^2)$ samples with overwhelming probability. This extends product constraints from hidden cuts of physical qubits~\cite{BoulandEtAl2025} to cuts within arbitrary Clifford frames. For Haar-unitary Choi states, we separately prove $\gamma_2(J(U))=1/4-o(1)$; hence pseudorandom unitaries must have maximal Choi width except with negligible probability (\SM{}~\ref{sm:haarunitary}).

\prlsection{Discussion}
We have shown that product quantum systems can be recovered even after a Clifford scrambles them globally. The resulting states can have extensive magic, volume-law entanglement and classical sampling hardness, but their quadratic Bell equations retain a complete record of the hidden product structure. This gives efficient tomography controlled by Clifford-product width and allows weak hidden entanglement to be neglected within the target error. 
The width measures the size of correlations that cannot be separated by a Clifford, rather than the total amount of magic or entanglement across physical cuts.

This distinction is particularly clear for Clifford-scrambled $T$-states $C|T\rangle^{\otimes n}$, whose width remains one while magic is extensive and entanglement can obey a volume law. The same structure emerges with exponentially high probability in random circuits with $T$ gates interspersed with independent uniform global Cliffords, at any fixed $T$-gate density below one. 
For unitary processes, the same structure allows us to learn any circuit with a single parallel $T$ layer between arbitrary Clifford circuits. We recover an implementation that acts correctly on every input, showing that an extensive $T$ count and unrestricted Clifford depth are compatible with efficient circuit learning.

The decoder also gives a constructive disentangler for Clifford-augmented MPS representations. It recovers the smallest possible largest product block directly from copies, complementing local optimization and circuit-aware algebraic methods~\cite{FuxEtAl2025,LiuClark2026,MasotLlimaEtAl2026}. Ordering the recovered blocks consecutively gives a residual MPS with bond dimension $\chi\le 2^{\lfloor\CPW(\psi)/2\rfloor}$. 
For a supplied Clifford-augmented MPS with residual bond dimension $\chi_{\rm in}$, Bell samples can be generated classically in $O(n\chi_{\rm in}^3+n^2)$ time per sample after preprocessing (\SM{}~\ref{sm:mps-bell}). %

Residual entanglement after Clifford optimization can quantify correlations beyond stabilizer simplification~\cite{HuangQianQin2025NsEE,BordaKuhlmannRincon2026}, such as the summed nonstabilizerness entanglement entropy (NsEE), which minimizes the sum of entanglement entropies over all consecutive cuts using one common Clifford. 
For pure states, this quantity is zero precisely when a global Clifford can turn the state into a product of single-qubit states, corresponding to $\CPW=1$ (\SM{}~\ref{sm:nsee}). For this width-one class, our support-gap assumption lets the algorithm find such a Clifford and attain the zero NsEE minimum. 

Our results also show that extensive magic alone is not enough for pseudorandomness. Earlier work established a linear $T$-gate requirement~\cite{GrewalEtAl2024Improved}. We impose an additional structural constraint: except with negligible probability, no nontrivial product cut may remain in any Clifford frame.

Our decoder is Clifford only and thus can be implemented without using magic. Further, across any fixed partition of the physical qubits, 
the learning can be done with local operations and classical communication (LOCC) only, making it suitable for distributed learning~\cite{LeoneEtAl2025LOCC}. The tomography step can be replaced by Clifford LOCC via measuring conjugated Pauli strings. %

Our hidden-cut distinguisher uses $O(n^2)$ copies and polynomial classical time, independently of the support gaps. For the uniform random Clifford orbit of $|T\rangle^{\otimes n}$, Bell estimates of squared Pauli correlations~\cite[Lemma~2.2]{king2025triply}, followed by a search over all Pauli labels, give an $O(n)$-copy distinguisher with exponential classical processing~\cite{bittel2026operational}. At least $\Omega(\sqrt n)$ copies are necessary for constant distinguishing advantage for this ensemble~\cite{bittel2026operational}. %

Finally, it would be interesting to extend our algorithm to nearly Pauli-stabilized subsystems, general Clifford-augmented MPSs whose residual states remain entangled across blocks, and noisy inputs near the promised class.

\begin{acknowledgments}
\begin{samepage}
\prlsection{Acknowledgments}
We thank Zhenhuan Liu for insightful discussions. The algorithm is implemented in Python on GitHub~\cite{haug2026productlearner}.
Generative artificial-intelligence tools were used to assist with derivations, language editing, organization, and manuscript preparation.
\par
\end{samepage}
\end{acknowledgments}

\bibliography{references}
\clearpage
\onecolumngrid

\let\addcontentsline\oldaddcontentsline

\appendix
\setcounter{section}{0}
\setcounter{equation}{0}
\setcounter{theorem}{0}
\setcounter{figure}{0}
\setcounter{table}{0}
\renewcommand{\thesection}{\Alph{section}}
\renewcommand{\thetheorem}{S\arabic{theorem}}
\renewcommand{\thefigure}{S\arabic{figure}}
\renewcommand{\theHtheorem}{supp.\arabic{theorem}}
\renewcommand{\theHsection}{supp.\Alph{section}}
\setcounter{secnumdepth}{2}
\begin{center}
{\large\bfseries Supplemental Material}
\end{center}
\noindent

\makeatletter
\@starttoc{toc}
\makeatother

\section{Mathematical background and notation}
\label{sm:math-background}

\subsection{System size and binary arithmetic}
\label{sm:notation}

We define $n$ as the number of qubits of the input state $\ket{\psi}$ and write
\begin{equation}
 \nu=n-\dim\mathsf S_\psi
 \label{sm:nullity-notation}
\end{equation}
for its stabilizer nullity, using the same symbol as in the main text. Thus $n-\nu$ is the number of independent Pauli stabilizers, and removing them leaves $\nu$ qubits with no nonidentity Pauli stabilizer. We use $\nu$ for the residual size throughout the structural and approximate-learning arguments. When a statement applies to a generic register, we denote its size by $N$.

Pauli labels belong to a vector space over $\F_2=\{0,1\}$, with addition and multiplication modulo two. In particular, $1+1=0$, addition is bitwise XOR, and subtraction equals addition. Matrix ranks, kernels, spans, and eigenvectors in the Pauli-label arguments are computed over this field. %

\subsection{Affine equations and Boolean quadratics}
\label{sm:boolean-background}

A Boolean function takes a binary vector $v=(v_1,\ldots,v_d)$ to one bit. Every such function has a unique multilinear polynomial representation over $\F_2$, because $v_i^2=v_i$ on binary inputs. We write $\RM(j,d)$ for the vector space of Boolean functions of degree at most $j$; this is the usual Reed--Muller notation. The cases needed here are
\begin{align}
 h(v)&=c+\sum_i\ell_i v_i,
 &h&\in\RM(1,d),\nonumber\\
 f(v)&=c+\sum_i\ell_i v_i+\sum_{i<j}b_{ij}v_i v_j,
 &f&\in\RM(2,d).
 \label{sm:boolean-expansion}
\end{align}
An affine function is a linear function plus a constant $c$. The equation $h(v)=0$ is a parity check: the selected bits sum to the prescribed value $c$. For example, $1+v_1+v_3=0$ requires $v_1$ and $v_3$ to differ. Similarly, an affine map $g(v)=Fv+t$ consists of a linear change of coordinates followed by a fixed translation $t$. Substitution by an invertible affine map preserves polynomial degree.

The support of a distribution $p$ is the set of labels with $p(v)>0$. An exact equation is a function that vanishes on this entire support. Its empirical counterpart vanishes on the observed samples only. Evaluating the monomials $1,v_i,v_iv_j$ on each sample gives a binary feature matrix. Multiplying this matrix by the coefficient vector of $f$ evaluates $f$ on every sample, so its nullspace consists exactly of the empirical equations. This explains why learning equations is a linear-algebra task even when the equations themselves are quadratic.

To extract the quadratic part, define the polarization
\begin{equation}
 \mathcal B_f(u,v)=f(u+v)+f(u)+f(v)+f(0)
 =\sum_{i<j}b_{ij}(u_i v_j+u_j v_i).
 \label{sm:polarization-background}
\end{equation}
All constant and linear terms cancel. Thus $\mathcal B_f$ is bilinear and can be written as $u^T B_fv$, where $(B_f)_{ij}=(B_f)_{ji}=b_{ij}$ and $(B_f)_{ii}=0$. Such a form is called alternating: $\mathcal B_f(v,v)=0$ for every $v$. Over $\F_2$, an alternating matrix is symmetric with zero diagonal; symmetry alone does not imply alternation. Two quadratics have the same polarization precisely when their difference is affine. For example, $f(x,z)=xz$ has
$B_f=\bigl(\begin{smallmatrix}0&1\\1&0\end{smallmatrix}\bigr)$, while adding $x$, $z$, or a constant leaves $B_f$ unchanged.

\subsection{Symplectic spaces and subsystem projectors}
\label{sm:symplectic-background}

On $N$ qubits, write a Pauli label as $v=(x,z)\in V=\F_2^{2N}$. The bilinear form
\begin{equation}
 [u,v]=u_x\cdot v_z+u_z\cdot v_x=u^TJv,
 \qquad J=\begin{pmatrix}0&I_N\\I_N&0\end{pmatrix}
 \label{sm:background-form}
\end{equation}
encodes commutation: two Paulis commute when $[u,v]=0$ and anticommute when $[u,v]=1$. This form is alternating and nondegenerate, meaning that a vector commuting with every label must be zero. A space equipped with such a form is called symplectic. Here $J^T=J$ and $J^2=I$.

For a subspace $U$, its symplectic orthogonal complement is
$U^\perp=\{v:[u,v]=0\text{ for all }u\in U\}$. A subspace is isotropic when its vectors commute pairwise, as happens for stabilizer labels. A subspace is nondegenerate when $U\cap U^\perp=\{0\}$. It then has even dimension and admits a symplectic basis $e_1,f_1,\ldots,e_t,f_t$ satisfying
\begin{equation}
 [e_i,e_j]=[f_i,f_j]=0,\qquad [e_i,f_j]=\delta_{ij}.
 \label{sm:symplectic-basis-background}
\end{equation}
Each pair can serve as the $X$ and $Z$ labels of one encoded qubit. Such a basis is constructed by choosing $e\ne0$, finding $f$ with $[e,f]=1$, and continuing in the orthogonal complement of their span. This is symplectic Gram--Schmidt. A full symplectic basis specifies a Clifford change of coordinates, up to Pauli phases. Consequently, a decomposition of the full Pauli-label space into mutually symplectically orthogonal, nondegenerate subspaces identifies quantum subsystems: a Clifford transformation maps each subspace to the Pauli labels of a disjoint group of qubits.

To convert a quadratic into a linear map, set $E_f=JB_f$. This definition is equivalent to
\begin{equation}
 [E_fu,v]=\mathcal B_f(u,v)=[u,E_fv].
 \label{sm:self-adjoint-background}
\end{equation}
The equality on the two sides is called self-adjointness with respect to the symplectic form; it means $E_f^TJ=JE_f$. %
A projector additionally satisfies $E^2=E$. It then fixes vectors in its image, annihilates vectors in its kernel, and gives the decomposition $v=Ev+(I+E)v$. If it is symplectically self-adjoint, its image and kernel are orthogonal nondegenerate spaces and therefore describe complementary subsystems. For a physical subset $A$, the quadratic $q_A(x,z)=\sum_{i\in A}x_iz_i$ gives exactly the projector retaining the $x_i,z_i$ coordinates for $i\in A$. Theorem~\ref{thm:structure} establishes the converse for exact Bell equations after stabilizer removal.

A family of commuting projectors can be diagonalized simultaneously by splitting into their joint zero and one eigenspaces. When it is also closed under addition and multiplication and contains $I$, it is a Boolean algebra of projectors. These operations represent symmetric difference, intersection, and, through $I+E$, complement of the selected subsystems. A minimal nonzero projector selects a subsystem that the algebra cannot split further. The classification proof shows why these algebraic pieces are precisely the irreducible product blocks of the state.

\subsection{Pauli expansions and reduced-state purity}
\label{sm:operator-background}

The Hermitian Paulis form an orthogonal operator basis:
\begin{equation}
 \Tr(W_uW_v)=2^N\delta_{uv},\qquad
 \rho=2^{-N}\sum_v\Tr(\rho W_v)W_v.
 \label{sm:pauli-expansion-background}
\end{equation}
Thus a Pauli commuting with every nonzero term in this expansion also commutes with $\rho$. For a rank-one state, commuting with $\rho=|\psi\rangle\langle\psi|$ means preserving its one-dimensional image, so $|\psi\rangle$ is a Pauli eigenvector. This is the link between the span of nonzero Pauli correlations and the absence of stabilizers used below.

Finally, a pure bipartite state has a Schmidt decomposition
$|\psi\rangle=\sum_i\sqrt{\lambda_i}|i\rangle_A|\widetilde i\rangle_{\bar A}$, with $\lambda_i\ge0$ and $\sum_i\lambda_i=1$. The nonzero eigenvalues of either reduced state are the $\lambda_i$. Its purity is $\Tr\rho_A^2=\sum_i\lambda_i^2$, which equals one exactly when a single Schmidt coefficient is nonzero, that is, when the state is a product across $A|\bar A$. We write $\|M\|_2^2=\Tr(M^\dagger M)$ for the Hilbert--Schmidt norm in the swap bound. %

\section{Bell measurements and stabilizer removal}
\label{sm:bell}

We now derive the connection between Bell samples, Pauli stabilizers, and product cuts. For an $N$-qubit system, use the conventions of Sec.~\ref{sm:math-background}: $v=(x,z)\in\F_2^{2N}$ and
\begin{equation}
 [u,v]=u^TJv,\qquad
 J=\begin{pmatrix}0&I_N\\ I_N&0\end{pmatrix},\qquad
 q(v)=x\cdot z.
 \label{sm:symplectic}
\end{equation}
The Hermitian Paulis are $W_v=\bigotimes_{j=1}^N i^{x_jz_j}X_j^{x_j}Z_j^{z_j}$, with ordinary integer exponents in the phase. Complex conjugation changes the sign of each $Y$ component, giving $W_v^*=W_v^T=(-1)^{q(v)}W_v$. We measure two identical copies in
the Bell basis. With $d=2^N$ and
$|\Phi_d\rangle=d^{-1/2}\sum_j|j,j\rangle$, the basis states and
sampling probabilities are
\begin{equation}
 |B_v\rangle=(I\otimes W_v)|\Phi_d\rangle,\qquad
 p_\psi(v)=d^{-1}|\langle\psi^*|W_v|\psi\rangle|^2.
 \label{sm:bell-distribution}
\end{equation}
We write $\xi\sim p_\psi$ for a random Bell sample and
$\xi^{(1)},\ldots,\xi^{(M)}$ for independent Bell samples.
To relate these samples to a product cut, consider a subset
$A\subseteq[N]$ and define $q_A(v)=\sum_{i\in A}x_iz_i$. The partial
swap is diagonal in the Bell basis. The relation
$\mathrm{SWAP}_A|B_v\rangle=(-1)^{q_A(v)}|B_v\rangle$ gives
\begin{equation}
 \E_{p_\psi}(-1)^{q_A(\xi)}=\Tr\rho_A^2,
 \qquad
 \Pr[q_A(\xi)=1]=\frac{1-\Tr\rho_A^2}{2}.
 \label{sm:swap-purity}
\end{equation}
Thus the frequency of $q_A=1$ measures the impurity of the reduced
state. For the full system, purity implies that every Bell label with nonzero probability satisfies the full-swap equation $q=0$.

To recover cuts hidden by a Clifford, we need its action on the measured
labels. A Clifford $C$ acts on Pauli labels by a symplectic matrix $F_C$,
whereas its action on Bell labels includes a translation:
\begin{equation}
 (C\otimes C)|B_v\rangle=e^{i\theta_v}|B_{g_C(v)}\rangle,
 \qquad g_C(v)=F_Cv+t_C.
 \label{sm:affine-action}
\end{equation}
To derive the shift, use $(A\otimes I)|\Phi_d\rangle=(I\otimes A^T)|\Phi_d\rangle$, so the matrix labeling the transformed Bell state is
\begin{equation}
 C W_v C^T=(CW_vC^\dagger)(CC^T).
 \label{sm:bell-shift-derivation}
\end{equation}
Complex conjugation changes only Pauli signs, so $C^*$ has the same symplectic action $F_C$ as $C$. Consequently $C^T=(C^*)^\dagger$ has action $F_C^{-1}$, and $CC^T$ has trivial symplectic action. A Clifford with trivial action on labels is a Pauli up to phase. Calling its label $t_C$, multiplication in Eq.~\eqref{sm:bell-shift-derivation} adds $t_C$ to $F_Cv$, giving the affine map in Eq.~\eqref{sm:affine-action}.

For example, the single-qubit phase gate $C=\operatorname{diag}(1,i)$ has $CC^T=Z$, so
\begin{equation}
 g_C(x,z)=(x,z+x+1).
 \label{sm:phase-gate-affine}
\end{equation}
In particular, it maps the Bell label $(0,0)$ to $(0,1)$; a linear map alone could not do this. Nevertheless, $q(g_C(v))=q(v)$ because $C\otimes C$ commutes with the full swap. Clifford decoding can therefore be applied directly to the samples by an affine relabeling. Under this invertible change of variables, a support equation becomes $f\circ g_C^{-1}$, with the same degree and violation probability in the transformed distribution.

We express an equation $f(v)=0$ as the Bell-diagonal observable
$Z_f=\sum_v(-1)^{f(v)}|B_v\rangle\langle B_v|$. Its expectation is $1-2\Pr[f(\xi)=1]$, so an exact equation has expectation one. To identify the affine equations, first note that
\begin{equation}
 (W_a\otimes W_a^*)|B_v\rangle=(-1)^{[a,v]}|B_v\rangle.
 \label{sm:affine-eigenvalue}
\end{equation}
Indeed, $(W_a\otimes W_a^*)|\Phi_d\rangle=|\Phi_d\rangle$, and commuting $W_a^*$ through the second-register $W_v$ supplies the sign $(-1)^{[a,v]}$. Since every linear functional is $[a,\cdot]$ for a unique $a$, every affine equation is $h(v)=c+[a,v]$. Equation~\eqref{sm:affine-eigenvalue} gives
\begin{equation}
 Z_h=(-1)^c W_a\otimes W_a^*,\qquad
 \langle Z_h\rangle_{\psi^{\otimes2}}
 =(-1)^{c+q(a)}|\langle W_a\rangle_\psi|^2.
 \label{sm:affine-character}
\end{equation}
For $a\ne0$, this expectation equals one precisely when $|\langle W_a\rangle_\psi|=1$ and $c=q(a)$. Since a Hermitian Pauli has eigenvalues $\pm1$, unit expectation magnitude means $W_a|\psi\rangle=\pm|\psi\rangle$. Thus nonzero exact affine equations are in one-to-one correspondence with nonidentity stabilizer labels. The case $a=0$ contributes only the zero function as an exact equation. The squared expectation in Eq.~\eqref{sm:affine-character} also explains why Bell samples do not reveal the stabilizer eigenvalue signs.

As a simple example, a Bell sample from $|0\rangle$ is $(x,z)=(0,0)$ or $(0,1)$, each with probability $1/2$. Its exact affine equation is $x=0$, corresponding to the stabilizer label of $Z$. The state $|1\rangle$ has the same Bell distribution, although its stabilizer is $-Z$. A measurement in the decoded computational basis supplies the missing sign.

Exact affine equations also generate quadratic equations that
need not describe product cuts. For example, for
$|\psi\rangle=|0\rangle\otimes|\phi\rangle$,
every Bell sample satisfies $x_1=0$. Hence
$f(v)=x_1x_2$ vanishes on the Bell support, independently
of $|\phi\rangle$. In coordinates $v=(x_1,x_2,z_1,z_2)$,
its polarization gives
\[
E_fv=(0,0,x_2,x_1).
\]
Thus $E_f$ has rank two but satisfies $E_f^2=0\ne E_f$:
it is not a projector and therefore does not specify
a product cut. Removing the stabilized qubit eliminates
this equation before we reconstruct the remaining subsystems.

\begin{lemma}[Stabilizer removal]
\label{sm:quotient}
Let $\mathsf S_\psi$ be the full Pauli-stabilizer label space of an
$n$-qubit pure state, with nullity $\nu=n-\dim\mathsf S_\psi$.
There is a Clifford decoder such that
\begin{equation}
 D_{\rm stab}^{\dagger}|\psi\rangle
 =|b\rangle_{1\ldots n-\nu}\otimes|\psi_0\rangle,
 \label{sm:quotient-state}
\end{equation}
where the residual state $|\psi_0\rangle$ has no nonidentity Pauli
stabilizer. If $|\psi\rangle\in\CTP_{n,k}$, then
$|\psi_0\rangle\in\CTP_{\nu,k}$. Relabeling Bell samples by
$g_{D_{\rm stab}^{\dagger}}$ and deleting the first $n-\nu$ qubit labels
gives Bell samples of $|\psi_0\rangle$ without postselection.
\end{lemma}

\begin{proof}
We separate the stabilized qubits by mapping a basis of their isotropic
label space to the labels of $Z_1,\ldots,Z_{n-\nu}$. Symplectic elimination
constructs this Clifford and gives Eq.~\eqref{sm:quotient-state}.
The residual state cannot have another Pauli stabilizer, since
$\mathsf S_\psi$ is complete. The Bell distribution factorizes over
the decoded blocks. Deleting the stabilized coordinates
therefore gives the residual Bell distribution.

We next show that this operation preserves the block-size promise.
Choose a hidden representation
$C^\dagger|\psi\rangle=\bigotimes_i|\phi_i\rangle$ with blocks of
size at most $k$. A product Pauli has expectation of magnitude one
exactly when each local Pauli does. The full stabilizer space is thus
the direct sum of the local stabilizer spaces. Within each block,
complete a basis of the isotropic stabilizer space to a symplectic
basis. The corresponding block-local Clifford removes its stabilized
qubits, leaving blocks of size at most $k$ and the same total number
$n-\nu$ of fixed qubits. After adjusting the stabilizer signs, both this
decoder and the learned decoder identify the original stabilizer code
with
$|0\rangle^{\otimes(n-\nu)}\otimes(\mathbb C^2)^{\otimes \nu}$. Their relative Clifford preserves this code. To identify its action on the residual qubits, let $S_0$ be the span of the fixed $Z$ labels. Paulis preserving the code have labels in $S_0^\perp$, and labels differing by an element of $S_0$ act identically on the code, up to a sign. The quotient $S_0^\perp/S_0$ is therefore the Pauli-label space of the $\nu$ residual qubits. It has dimension $2\nu$ and inherits a nondegenerate symplectic form. The relative Clifford induces a symplectic map on this quotient and hence a Clifford on the residual register. Thus the two residual states differ by a $\nu$-qubit Clifford, and the learned residual state also belongs to $\CTP_{\nu,k}$.
\end{proof}

To implement this first stage, we use the standard affine-nullspace
method for stabilizer learning~\cite{montanaro2017learning,GrewalEtAl2025,HinscheEtAl2026}.
We form a binary matrix with rows $(1,(\xi^{(t)})^T)$ for $t=1,\ldots,M$.
Each nullspace vector $(c,\ell)$ specifies an affine equation
$c+\ell^Tv=0$ and hence a candidate stabilizer label
$a=J\ell$, by Eq.~\eqref{sm:affine-character}.
Measuring the fixed qubits of one fresh decoded copy determines the
signs $b$. We return these qubits as single-qubit blocks and learn the
remaining subsystems from the residual samples. If $\nu=0$, the entire state is stabilized and there is no residual reconstruction problem. Otherwise, the following sections use $\nu$ for the residual number of qubits. If the state is supplied as an MPS, the samples can be generated classically by the sampler in Sec.~\ref{sm:mps-bell}.

\section{Every exact quadratic equation is a product cut}
\label{sm:classification}

We now prove Theorem~\ref{thm:structure}. Throughout this section,
$|\psi\rangle$ is the stabilizer-free residual state on $\nu$ qubits,
where $\nu$ is the stabilizer nullity of the original state. We use the
binary linear algebra introduced in Sec.~\ref{sm:math-background}.
The proof has three steps. First, an exact quadratic Bell equation
gives a projector on the Pauli-label space. Second, the image of that
projector corresponds to a subsystem with a pure reduced state, and
therefore to a product cut. Finally, all such cuts are compatible;
their common refinement gives the unique irreducible decomposition.

Let $V=\F_2^{2\nu}$ and let
$K_\psi\subseteq\RM(2,2\nu)$ be the vector space of Boolean
polynomials of degree at most two that vanish on
$\supp p_\psi$. For $f\in\RM(2,2\nu)$, define its polarization
matrix $B_f$ and the associated linear map $E_f$ by
\begin{equation}
 u^TB_fv=f(u+v)+f(u)+f(v)+f(0),\qquad E_f=JB_f.
 \label{sm:polarization}
\end{equation}
The right-hand side keeps only the quadratic coefficients of $f$.
It is bilinear, symmetric, and zero when $u=v$, so $B_f$ is an
alternating matrix. Since $B_f^T=B_f$ and $J^2=I$, the map satisfies
\begin{equation}
 E_f^TJ=B_f=JE_f,\qquad [E_fu,v]=[u,E_fv].
 \label{sm:polarization-self-adjoint}
\end{equation}
Thus $E_f$ is self-adjoint with respect to the symplectic form.
For the partial-swap polynomial $q_A$, it is the projector
$E_{q_A}=P_A$ onto the Pauli coordinates of $A$.

This construction respects a change of Clifford frame. Under
the affine relabeling $g_C(v)=F_Cv+t_C$ from
Eq.~\eqref{sm:affine-action}, the translation changes only the
constant and linear terms. Hence
\begin{equation}
 B_{f\circ g_C^{-1}}=F_C^{-T}B_fF_C^{-1},\qquad
 E_{f\circ g_C^{-1}}=F_CE_fF_C^{-1}.
 \label{sm:polarization-covariance}
\end{equation}
Polarization therefore turns a scrambled partial swap into the
corresponding conjugated projector. It also distinguishes exact
equations: if $E_f=E_g$ for $f,g\in K_\psi$, then $f+g$ is an
affine equation vanishing on the Bell support. Stabilizer removal
eliminates every nonzero such equation by
Eq.~\eqref{sm:affine-character}, so $f=g$. The map
$f\mapsto E_f$ is consequently injective on $K_\psi$.

\subsection{An exact quadratic equation gives a projector}

\begin{lemma}[Projectors from quadratic equations]
\label{sm:rigidity}
Let $|\psi\rangle$ be a stabilizer-free pure state. Every
$f\in K_\psi$ has idempotent polarization, $E_f^2=E_f$.
\end{lemma}

\begin{proof}
We first derive an identity for Pauli expectation values. Set
$E=E_f$, $A_v=W_v\otimes I$, and
$R_\psi(v)=|\langle W_v\rangle_\psi|$.
The operator $A_v$ maps each Bell vector $|B_y\rangle$ to a
phase times $|B_{y+v}\rangle$. Applying $Z_f$, then $A_v$,
then $Z_f$ to this vector multiplies that phase by
$(-1)^{f(y)+f(y+v)}$. Therefore
\begin{equation}
 \begin{split}
 Z_fA_vZ_f&=A_vZ_{h_v},\\
 h_v(y)&=f(y+v)+f(y)=c_v+[Ev,y],\\
 c_v&=f(v)+f(0).
 \end{split}
 \label{sm:conjugation}
\end{equation}
The second line follows directly from polarization:
$f(y+v)+f(y)=f(v)+f(0)+v^TB_fy$, with
$v^TB_fy=[Ev,y]$. In particular, this finite difference is
affine, even when $f$ is quadratic.

Equation~\eqref{sm:affine-character} now gives
$Z_{h_v}=(-1)^{c_v}W_{Ev}\otimes W_{Ev}^*$.
Multiplying its first Pauli by $W_v$ adds their binary labels,
so Eq.~\eqref{sm:conjugation} becomes
\begin{equation}
 Z_f(W_v\otimes I)Z_f
 =\omega_v W_{(I+E)v}\otimes W_{Ev}^*,
 \qquad |\omega_v|=1.
 \label{sm:conjugation-paulis}
\end{equation}
In fact, $[v,Ev]=v^TB_fv=0$, so the two Paulis in the first
copy commute and $\omega_v$ is a real sign. This sign is
immaterial below. Because $f$ vanishes on the full
Bell support, $Z_f|\psi\rangle^{\otimes2}=|\psi\rangle^{\otimes2}$.
Taking the expectation of Eq.~\eqref{sm:conjugation-paulis}, and
then its absolute value, yields
\begin{equation}
 R_\psi(v)=R_\psi(Ev)R_\psi((I+E)v).
 \label{sm:rigidity-identity}
\end{equation}
Here $W_a^*=(-1)^{q(a)}W_a$, so complex conjugation does not
change the absolute value of the Pauli expectation. The two-copy
expectation factorizes because the measured state is
$|\psi\rangle\otimes|\psi\rangle$.

We next show that this identity forces $E^2=E$ on every label
with nonzero Pauli expectation. Define
\begin{equation}
 \Delta=E^2+E,\qquad
 T=\{v\in V:R_\psi(v)>0\}.
 \label{sm:idempotence-defect}
\end{equation}
Since addition and subtraction agree over $\F_2$, idempotence is
equivalent to $\Delta=0$. Suppose that $\Delta v\ne0$ for some
$v\in T$, and choose such a label maximizing $R_\psi(v)$.
This maximum exists because the label space is finite.
Neither $Ev$ nor $(I+E)v$ can be zero. Indeed, $Ev=0$ gives
$E^2v=0$, whereas $(I+E)v=0$ gives $Ev=v$ and then $E^2v=v$;
both cases would imply $\Delta v=0$.

Equation~\eqref{sm:rigidity-identity} implies that both Pauli
expectations on its right-hand side are positive. They are also
strictly below one: a Hermitian Pauli has expectation of magnitude
one on a pure state only when that state is its eigenvector, and
the residual state has no nonidentity Pauli stabilizer. Consequently,
\begin{equation}
 R_\psi(Ev)>R_\psi(v),\qquad
 R_\psi((I+E)v)>R_\psi(v).
 \label{sm:strict-correlation-increase}
\end{equation}
At least one of these two labels still violates idempotence,
because linearity gives
$\Delta(Ev)+\Delta((I+E)v)=\Delta v\ne0$.
It belongs to $T$ and has larger expectation than the chosen
maximizer, a contradiction. Thus $\Delta$ vanishes on $T$.

To extend this conclusion to every label, we prove that $T$ spans
$V$. Write the Pauli expansion of the rank-one state as
\begin{equation}
 \rho=|\psi\rangle\langle\psi|
 =2^{-\nu}\sum_{v\in T}\langle W_v\rangle_\psi W_v.
 \label{sm:pauli-support-expansion}
\end{equation}
If $\Span T$ were a proper subspace, nondegeneracy of the
symplectic form would give a nonzero
$a\in(\Span T)^\perp$. Then $[a,v]=0$ for every $v\in T$,
so $W_a$ commutes with every term in
Eq.~\eqref{sm:pauli-support-expansion}, and hence with $\rho$.
Commutation with a rank-one projector preserves its
one-dimensional image, implying
$W_a|\psi\rangle=\pm|\psi\rangle$. This is a forbidden
nonidentity Pauli stabilizer. Therefore $\Span T=V$.
Since the linear map $\Delta$ vanishes on a spanning set,
$\Delta=0$ on $V$, as claimed.
\end{proof}

\subsection{The projector identifies a product subsystem}

An idempotent linear map separates a vector space into its image
and kernel. Self-adjointness gives the additional property needed
here: these two spaces describe commuting sets of Pauli operators.
For $E^2=E$ with $E^TJ=JE$, every $v\in V$ has the unique
decomposition
\begin{equation}
 v=Ev+(I+E)v,\qquad
 V=\operatorname{im}E\oplus\ker E,
 \qquad
 \ker E=(\operatorname{im}E)^\perp.
 \label{sm:symplectic-projector-split}
\end{equation}
For example, if $u=Ex$ and $w\in\ker E$, then
$[u,w]=[Ex,w]=[x,Ew]=0$. Conversely, if $w$ is orthogonal
to $\operatorname{im}E$, then $[x,Ew]=0$ for all $x$, forcing
$Ew=0$. The restrictions of the symplectic form to the image
and kernel are nondegenerate: a vector in either space that is
orthogonal to that whole space is orthogonal to both spaces,
and must be zero.

Both spaces therefore have even dimension and admit symplectic
bases. Choose one within each space and concatenate them. The
resulting change of Pauli coordinates is implemented by a Clifford,
and sends $E$ to $P_A$ for a subsystem of
$|A|=\rank(E)/2$ qubits. This identifies a candidate subsystem;
we must still prove that its reduced state is pure. The following
bound does so and also provides the quantitative statement needed
for approximate learning.

\begin{lemma}[Purity bound for swaps with Pauli operators]
\label{sm:twisted-swap}
Let $\rho$ be pure on $A\bar A$. For any Pauli $W_a$,
\begin{equation}
 \left|\Tr[(\rho\otimes\rho)\mathrm{SWAP}_A
       (W_a\otimes W_a^*)]\right|\le\Tr\rho_A^2.
 \label{sm:twisted-bound}
\end{equation}
If $E_f^2=E_f$ and
$\Pr_{p_\psi}[f(\xi)=1]\le\tau\le1/2$, the subsystem associated
with $\operatorname{im}E_f$ consequently has purity at least
$1-2\tau$ in an adapted Clifford frame.
\end{lemma}

\begin{proof}
Write $W_a=R_A\otimes R_{\bar A}$ and define
$M=\Tr_{\bar A}[(I\otimes R_{\bar A})\rho]$.
The matrix $M$ is Hermitian: partial trace is cyclic for operators
acting only on the traced subsystem, so
$M^\dagger=\Tr_{\bar A}[\rho(I\otimes R_{\bar A})]=M$.
Since $W_a^*=(-1)^{q(a)}W_a$, the complex conjugation contributes
only an overall sign. Taking the partial trace over both copies
of $\bar A$ and using
$\Tr[(X\otimes Y)\mathrm{SWAP}_A]=\Tr(XY)$ gives
\begin{equation}
 \begin{split}
 \left|\Tr[(\rho\otimes\rho)\mathrm{SWAP}_A
       (W_a\otimes W_a^*)]\right|
 &=\left|\Tr[(R_AM)^2]\right|\\
 &\le\Tr[(R_AM)^\dagger(R_AM)]
 =\|M\|_2^2.
 \end{split}
 \label{sm:twisted-trace-bound}
\end{equation}
The inequality is the Hilbert--Schmidt Cauchy--Schwarz inequality,
and the last equality uses $R_A^\dagger R_A=I$.

To bound $\|M\|_2^2$, write a Schmidt decomposition
$|\psi\rangle=\sum_i\sqrt{\lambda_i}|i\rangle_A|i\rangle_{\bar A}$,
where $\lambda_i\ge0$ and $\sum_i\lambda_i=1$. Its partial trace gives
\begin{equation}
 M=\sum_{i,j}\sqrt{\lambda_i\lambda_j}
 \langle j|R_{\bar A}|i\rangle\,|i\rangle\langle j|.
 \label{sm:schmidt-matrix}
\end{equation}
Thus, with $D_{ij}=|\langle j|R_{\bar A}|i\rangle|^2$,
\begin{equation}
 \begin{split}
 \|M\|_2^2
 &=\sum_{i,j}\lambda_i\lambda_jD_{ij}\\
 &\le\frac12\sum_{i,j}(\lambda_i^2+\lambda_j^2)D_{ij}
 =\sum_i\lambda_i^2=\Tr\rho_A^2.
 \end{split}
 \label{sm:schmidt-bound}
\end{equation}
For the last equality, extend the Schmidt vectors to a complete
basis of $\bar A$ and pad the list of $\lambda_i$ with zeros.
Unitarity of $R_{\bar A}$ then ensures that every row and column
of $D$ sums to one. This step does not assume that the Schmidt
support is invariant under $R_{\bar A}$. Equations~\eqref{sm:twisted-trace-bound}
and \eqref{sm:schmidt-bound} prove Eq.~\eqref{sm:twisted-bound}.

Now suppose $E_f$ is idempotent. Use
Eq.~\eqref{sm:symplectic-projector-split} to choose a Clifford
frame in which $E_f=P_A$, and relabel $f$ accordingly.
The polynomials $f$ and $q_A$ have the same polarization,
so their difference is affine. Every linear functional on $V$
has the form $[a,\cdot]$ for a unique $a$, giving
$f=q_A+c+[a,\cdot]$. Since the corresponding observables are
diagonal in the same Bell basis,
\begin{equation}
 Z_f=(-1)^c\mathrm{SWAP}_A(W_a\otimes W_a^*).
 \label{sm:twisted-operator}
\end{equation}
The violation probability is preserved by Clifford relabeling,
and therefore
\begin{equation}
 1-2\tau\le\langle Z_f\rangle
 =1-2\Pr[f(\xi)=1]
 \le|\langle Z_f\rangle|\le\Tr\rho_A^2.
 \label{sm:projector-purity}
\end{equation}
For an exact equation, $\tau=0$, so $\Tr\rho_A^2=1$.
The Schmidt coefficients then obey
$\sum_i\lambda_i=\sum_i\lambda_i^2=1$, which is possible
only when exactly one coefficient is nonzero. Hence
$|\psi\rangle=|\psi_A\rangle\otimes|\psi_{\bar A}\rangle$.
This proves that the candidate subsystem is a product subsystem.
\end{proof}

For an exact equation in this frame, the affine term also vanishes.
Indeed, the product cut makes $q_A$ an exact Bell equation, so
$f+q_A$ is an exact affine equation. Stabilizer freedom forces
$f=q_A$. Thus an exact quadratic equation specifies the partial
swap itself after a suitable Clifford change of coordinates.

\subsection{Compatible cuts and the irreducible decomposition}

\begin{theorem}[Boolean algebra of exact cuts]
\label{sm:boolean}
For a stabilizer-free pure state, the space of polarizations
$\cL_\psi=\{E_f:f\in K_\psi\}$ is a Boolean algebra of symplectic
projectors. Its minimal nonzero projectors $P_1,\ldots,P_r$ are
unique up to permutation and satisfy
\begin{equation}
 \cL_\psi=\bigoplus_{\alpha=1}^r\F_2P_\alpha,
 \quad P_\alpha P_\beta=\delta_{\alpha\beta}P_\alpha,
 \quad\sum_\alpha P_\alpha=I.
 \label{sm:boolean-algebra}
\end{equation}
The projector images give the unique symplectic decomposition into
irreducible blocks. In the corresponding Clifford coordinates,
\begin{equation}
 K_\psi=\Span\{q_{A_1},\ldots,q_{A_r}\},\qquad
 \dim K_\psi=r,\qquad
 \CPW(\psi)=\max_\alpha\frac{\rank P_\alpha}{2}.
 \label{sm:exact-width}
\end{equation}
\end{theorem}

Here a nonzero projector is minimal if no other nonzero projector in
the algebra has an image strictly contained in its image. Such a
projector describes a block that no further exact cut can divide.

\begin{proof}
We first establish that any two exact cuts commute. The map
$f\mapsto E_f$ is linear, so $\cL_\psi$ is a vector space.
For $E,F\in\cL_\psi$, Lemma~\ref{sm:rigidity} applies to
$E$, $F$, and $E+F$. Expanding the last idempotence relation gives
\begin{equation}
 (E+F)^2=E+F+EF+FE=E+F,
 \qquad EF=FE.
 \label{sm:commuting-cuts}
\end{equation}
The full-swap equation $q\in K_\psi$ gives $I\in\cL_\psi$.
Thus the complement $I+E$ of each exact cut is also in the space.

Commutation allows simultaneous splitting into physical subsystems.
For two projectors, define their joint eigenspaces by
\begin{equation}
 V_{ab}=\{v\in V:Ev=av,\ Fv=bv\},\qquad a,b\in\F_2.
 \label{sm:joint-cut-spaces}
\end{equation}
Because $E$ and $F$ commute, $F$ preserves both the image and
kernel of $E$. Splitting these two spaces once more with $F$
gives $V=V_{11}\oplus V_{10}\oplus V_{01}\oplus V_{00}$.
Distinct summands are symplectically orthogonal: if two vectors
have different eigenvalues under either self-adjoint projector,
Eq.~\eqref{sm:polarization-self-adjoint} forces their pairing
to vanish. Each nonzero summand is consequently nondegenerate.
Choosing symplectic bases identifies them with four disjoint
qubit subsystems in a single Clifford frame.

We next prove closure under multiplication; commutation alone
does not yet establish this property. Lemma~\ref{sm:twisted-swap}
shows that the state factors across the $E$ cut, so in the joint
frame it can be written as
\begin{equation}
 |\psi\rangle
 =|\alpha\rangle_{11,10}\otimes|\beta\rangle_{01,00}.
 \label{sm:two-cut-product}
\end{equation}
The $F$ cut is also a product cut, so its $F=1$ reduction is pure.
Using Eq.~\eqref{sm:two-cut-product}, this reduction is
$\rho_{11}\otimes\rho_{01}$, whose purity is
$\Tr\rho_{11}^2\Tr\rho_{01}^2=1$.
Each purity is at most one, so both equal one. Purity of
$|\alpha\rangle$ then makes it a product on $11|10$, and
purity of $|\beta\rangle$ makes it a product on $01|00$.
The state therefore factors over all four joint subsystems
(with zero-dimensional label spaces omitted).

In particular, $V_{11}$ is a product subsystem, and $EF$ is
its projector. Its partial-swap polynomial is an exact Bell
equation in the joint frame. Pulling that equation back through
the affine Clifford relabeling gives an element of $K_\psi$
with polarization $EF$. Thus $EF\in\cL_\psi$. We have proved
that $\cL_\psi$ contains $I$, is closed under addition and
multiplication, and consists of commuting idempotents, which
is the Boolean-algebra statement.

To describe all its blocks explicitly, choose a vector-space
basis $E_1,\ldots,E_\ell$ of $\cL_\psi$. For every binary
string $\boldsymbol\epsilon\in\F_2^\ell$, form
\begin{equation}
 P_{\boldsymbol\epsilon}
 =\prod_{j=1}^{\ell}
 \begin{cases}
 E_j,&\epsilon_j=1,\\
 I+E_j,&\epsilon_j=0.
 \end{cases}
 \label{sm:joint-projectors}
\end{equation}
These products belong to the algebra by the closure just proved.
They project onto the joint eigenspaces with eigenvalues
$\boldsymbol\epsilon$. Two different products multiply to zero,
because at an index where their strings differ the product
contains $E_j(I+E_j)=0$. Expanding
$\prod_j(E_j+(I+E_j))=I$ shows that they sum to the identity.
Discard the zero products and denote the rest by
$P_1,\ldots,P_r$.

Every $E\in\cL_\psi$ is a linear combination of the $E_j$,
so on each of these joint eigenspaces it acts as either zero
or the identity. Consequently every $E$ is a sum of a subset
of the $P_\alpha$. No nonzero element of the algebra can
split the image of a $P_\alpha$, proving minimality.
Conversely, a sum containing two or more $P_\alpha$ is not
minimal. This characterizes the minimal projectors without
reference to the chosen basis $E_j$, and proves their
uniqueness up to permutation. Their mutually disjoint images
also make the projectors linearly independent, giving
Eq.~\eqref{sm:boolean-algebra}.

Each $P_\alpha$ defines an exact product cut. Successively
factoring these disjoint cuts gives
$|\psi\rangle=\bigotimes_\alpha|\phi_\alpha\rangle_{A_\alpha}$
in their common Clifford frame. A block cannot have a further
Clifford product decomposition: its cut projector would extend
by zero on the other blocks to a nonzero element of
$\cL_\psi$ strictly below $P_\alpha$, contradicting minimality.
The blocks are therefore irreducible.

Conversely, any Clifford product decomposition supplies
partial-swap equations and hence projectors in this same
algebra. Since each such projector is a sum of the
$P_\alpha$, every alternative decomposition groups these
irreducible symplectic subspaces. Its largest block must
therefore be at least as large as the largest irreducible block,
while the irreducible decomposition itself attains that size.
An image of dimension $\rank P_\alpha$ describes
$\rank P_\alpha/2$ qubits, proving the width formula in
Eq.~\eqref{sm:exact-width}.

Finally, each block-swap equation $q_{A_\alpha}$ lies in
$K_\psi$ and has polarization $P_\alpha$. These projectors
span $\cL_\psi$, and polarization is injective on $K_\psi$.
Thus the block-swap equations span all of $K_\psi$ and are
linearly independent. This proves the remaining identities
in Eq.~\eqref{sm:exact-width}.
\end{proof}

Together, Lemmas~\ref{sm:rigidity} and \ref{sm:twisted-swap}
and Theorem~\ref{sm:boolean} prove Theorem~\ref{thm:structure}
of the main text. The unique objects are the symplectic subspaces
of the irreducible blocks. A choice of Pauli coordinates within
each subsystem, the corresponding Clifford phases, and the
order of the blocks remain free.

For the original $n$-qubit state, Lemma~\ref{sm:quotient}
separates $n-\nu$ stabilized qubits from the $\nu$-qubit residual
state. Restoring the stabilized qubits as single-qubit blocks
gives width $\max\{1,\CPW(\psi_0)\}$ when any stabilized
qubits are present. This is optimal because stabilizer removal
preserves any block-size promise. If $\nu=0$, the state is fully
stabilized and has width one; if $\nu=n$, its width is the
residual width derived above.

\section{Exact reconstruction from finitely many samples}
\label{sm:exact-learning}

We now determine how many Bell samples suffice to recover the exact
equations and construct the decoder. The relevant gap is the smallest
probability that a false equation is violated. For $j=1,2$, let
$K_\psi^{(j)}$ be the exact degree-at-most-$j$ kernel and define
\begin{equation}
 \gamma_j(\psi)=\min_{f\in\RM(j,2n)\setminus K_\psi^{(j)}}
                    \Pr_{p_\psi}[f(\xi)=1].
 \label{sm:gap}
\end{equation}
Affine Clifford relabeling preserves both gaps, and
$\gamma_1\ge\gamma_2$. The affine gap has a direct Pauli
interpretation: for a stabilizer-free residual state on $\nu>0$ qubits,
\begin{equation}
 \gamma_1(\psi)=\frac{1-\max_{v\ne0}R_\psi(v)^2}{2}.
 \label{sm:affine-gap}
\end{equation}
Thus it measures separation from an additional Pauli stabilizer.
The gap definitions impose no lower bound on individual nonzero Bell
probabilities.

To find the equations, we evaluate all multilinear monomials of degree
at most $j$ on the measured samples. There are $D_1(n)=2n+1$ affine
features and $D_2(n)=1+2n+\binom{2n}{2}$ quadratic features. The
binary nullspace of the resulting matrix is the empirical kernel
$K_M^{(j)}$. Every true equation belongs to this kernel. If
$\gamma_j(\psi)\ge\gamma>0$, a fixed false equation holds on all $M$
samples with probability at most $e^{-M\gamma}$. A union bound over
the $2^{D_j(n)}$ Boolean functions gives
\begin{equation}
 \Pr[K_M^{(j)}\ne K_\psi^{(j)}]
 \le 2^{D_j(n)}e^{-M\gamma}.
 \label{sm:exact-kernel-prob}
\end{equation}
The kernel is therefore exact with probability at least $1-\delta$
once $M\ge[D_j(n)\ln2+\ln(1/\delta)]/\gamma$. Each Bell sample
uses two copies of the state.

The decoder is chosen from the data, so reusing the same samples after
decoding requires a guarantee that holds uniformly over possible
choices. Exact kernel recovery supplies this guarantee.

\begin{lemma}[Preservation under adaptive decoding]
\label{sm:inheritance}
Suppose the original degree-$j$ empirical kernel is exact. After any
learned Clifford relabeling and projection onto an exact pure block,
the degree-$j$ empirical kernel of the projected samples is also exact.
The block state's  degree-$j$ gap is at least that of the original state.
\end{lemma}

\begin{proof}
We compare each projected equation with an equation on the original
samples. Extend its polynomial to be independent of the deleted
coordinates, and pull it back through the affine Clifford action.
This preserves its degree and both its empirical and true
violation probabilities. Exactness of the original kernel then implies
exactness of the projected kernel. Every positive violation probability
is at least the original gap, which also proves the gap bound. Since
the original event holds for all polynomials simultaneously, the
decoder and retained subsystem can be chosen using these same samples.
\end{proof}

\subsection{Explicit reconstruction from Bell samples}
\label{sm:algorithm}

We now give a concrete implementation. All matrix operations in Algorithm~S1 are over $\F_2$. Write $\mathcal N(A)$ for a matrix whose columns form a basis of $\ker A$, obtained by Gaussian elimination, and $J_d=\left(\begin{smallmatrix}0&I_d\\I_d&0\end{smallmatrix}\right)$. The algorithm receives only the Bell samples; it does not require the preparation circuit or the block states. Its structural guarantee holds when the empirical degree-two kernel is exact, as ensured with high probability by Eq.~\eqref{sm:exact-kernel-prob} under the support-gap assumption.

\medskip
\noindent\textbf{Algorithm S1: Recover the Clifford frame and product cuts.}

\noindent\textbf{Input:} $M$ Bell samples $\xi^{(t)}=(x^{(t)},z^{(t)})\in\F_2^{2n}$, with all $x$ coordinates followed by all $z$ coordinates.\par
\noindent\textbf{Output:} a decoder $D^\dagger$, the nullity $\nu$, and decoded blocks $A_\alpha$, together with $n-\nu$ stabilized single-qubit blocks. If an algebraic consistency check fails, return \emph{insufficient data}.\par

\begin{enumerate}[leftmargin=*,itemsep=5pt]
 \item \textbf{Recover candidate stabilizer labels.}
 Form $\Phi_1$ with row $t$ equal to $(1,(\xi^{(t)})^T)$. For each column $(c_i,\ell_i)$ of $\mathcal N(\Phi_1)$, set $a_i=J_n\ell_i$. Let $m=\dim\ker\Phi_1$ and $\nu=n-m$. Check $m\le n$, $[a_i,a_j]=0$ for all $i,j$, and $c_i=q(a_i)$ for every $i$. The labels $a_i$ span the candidate Pauli-stabilizer space.

 \item \textbf{Remove stabilized coordinates.}
 By symplectic elimination, construct $G_{\rm stab}$ mapping this span to the labels of $Z_1,\ldots,Z_m$; explicit gate operations are given below. For each sample, compute $g_{G_{\rm stab}}(\xi^{(t)})$ using Eq.~\eqref{sm:affine-action} and retain only the $x,z$ coordinates of qubits $m+1,\ldots,n$, obtaining $v^{(t)}\in\F_2^{2\nu}$. If $\nu=0$, return $D^\dagger=G_{\rm stab}$ and $n$ stabilized single-qubit blocks.

 \item \textbf{Find the residual quadratic equations.}
 Form the feature matrix $\Phi_2$ with rows
 \[
  \left(1,v_1^{(t)},\ldots,v_{2\nu}^{(t)},
        (v_i^{(t)}v_j^{(t)})_{1\le i<j\le2\nu}\right).
 \]
 A column $c$ of $\mathcal N(\Phi_2)$ specifies
 $f_c(v)=c_0+\sum_i c_i v_i+\sum_{i<j}c_{ij}v_iv_j$.
 Set $(B_c)_{ij}=(B_c)_{ji}=c_{ij}$ and $(B_c)_{ii}=0$, and compute $E_c=J_\nu B_c$. Denote the resulting maps by $E_1,\ldots,E_\ell$.

 \item \textbf{Check the candidate projector algebra.}
 If $\ell>\nu$, return \emph{insufficient data} before checking products of the maps. Otherwise, check that the $E_i$ are linearly independent, that $E_i^2=E_i$, that $E_iE_j=E_jE_i$, and that $I_{2\nu}\in\Span\{E_i\}$. These checks hold for the exact kernel by Theorem~\ref{sm:boolean}, which also gives $\ell\le\nu$. Self-adjointness follows directly from $B_c^T=B_c$.

 \item \textbf{Split into joint eigenspaces.}
 Initialize a list of column-basis matrices $\mathcal V=\{I_{2\nu}\}$. For $i=1,\ldots,\ell$, replace every $V\in\mathcal V$ by the nonzero-column matrices
 \[
   V\,\mathcal N\!\left((E_i+\lambda I_{2\nu})V\right),
   \qquad\lambda\in\{0,1\}.
 \]
 After the last map, write $\mathcal V=\{V_1,\ldots,V_r\}$. Check $r=\ell$ and that each $V_\alpha$ has even column dimension and nondegenerate restricted form $V_\alpha^TJ_\nu V_\alpha$. Set $b_\alpha=\rank V_\alpha/2$.

 \item \textbf{Choose Pauli coordinates within each block.}
 Apply symplectic Gram--Schmidt to the columns of each $V_\alpha$: choose a nonzero vector $a$, find $b$ in the current space with $[a,b]=1$, retain this pair, and replace every remaining vector $u$ by
 \[
   u+[u,b]a+[u,a]b.
 \]
 Keep a basis of the resulting space and repeat until it is zero. Concatenate the pairs block by block, yielding $a_1,\ldots,a_\nu,b_1,\ldots,b_\nu$, and set $F=(a_1\cdots a_\nu\ b_1\cdots b_\nu)$. Then $F^TJ_\nu F=J_\nu$.

 \item \textbf{Synthesize the residual Clifford decoder.}
 Construct $G_{\rm res}$ with unsigned symplectic action $F^{-1}$, so
 \[
  G_{\rm res}W_{a_j}G_{\rm res}^\dagger=\pm X_j,
  \qquad
  G_{\rm res}W_{b_j}G_{\rm res}^\dagger=\pm Z_j.
 \]
 A gate synthesis is described below. The signs can be chosen through any valid Clifford lift; local Pauli differences do not change the cuts. Assign consecutive sets $\widetilde A_\alpha$ of $b_\alpha$ residual qubits in the same order as the pairs.

 \item \textbf{Return the decoded blocks.}
 Set $D^\dagger=(I_m\otimes G_{\rm res})G_{\rm stab}$ and $A_\alpha=\{m+j:j\in\widetilde A_\alpha\}$. Return these blocks and the stabilized singletons $\{1\},\ldots,\{m\}$. On exact recovery,
 \[
  D^\dagger|\psi\rangle
   =|b\rangle_{1\ldots m}\otimes\bigotimes_{\alpha=1}^r|\phi_\alpha\rangle_{A_\alpha},
  \qquad
  \CPW(\psi)=\max\{1,b_1,\ldots,b_r\}.
 \]
 Optionally, the value of $b\in\{0,1\}^{m}$ can be read out by measuring the first $m$ qubits of one fresh decoded copy. 
 Finally, learning the block states requires the subsequent tomography described in Sec.~\ref{sm:tomography-validation}.
\end{enumerate}

\noindent\emph{Affine relabeling in step 2.---}
The samples must be transformed with the affine, rather than just the symplectic, Clifford action. Gate by gate, $H_i$ swaps $x_i,z_i$; $S_i$ and $S_i^\dagger$ give $z_i\leftarrow z_i+x_i+1$; and $\mathrm{CNOT}_{i\to j}$ gives $x_j\leftarrow x_j+x_i$, $z_i\leftarrow z_i+z_j$. SWAP exchanges both pairs of coordinates. For ordinary Pauli labels the same rules apply without the constant $1$ in the phase-gate update. These rules allow every sample to be relabeled directly from the gate list of $G_{\rm stab}$.

\noindent\emph{Gate synthesis in steps 2 and 7.---}
To isolate the stabilizers, process their independent labels in order. At step $i$, add previously processed $Z_j$ generators to remove any $Z_j$ components with $j<i$. Commutation already excludes $X_j$ components there. On the remaining sites, use $H$ to turn $X$ into $Z$ and $S$ followed by $H$ to turn $Y$ into $Z$, ignoring signs. Swap a site in the support to site $i$. For each other $Z_j$ in the support, apply $\mathrm{CNOT}_{j\to i}$ to cancel it. Update every generator after each gate. This maps the current generator to $Z_i$ while preserving the previously isolated generators.

For the residual frame, update every column of $F$ after each gate and process its symplectic pairs in order. To map the current $a_i$ to $X_i$, turn each supported $Z$ into $X$ with $H$ and each $Y$ into $X$ with $S$. Swap one supported site to $i$, then clear the other $X_j$ components with $\mathrm{CNOT}_{i\to j}$. Its partner $b_i$ now has $z_i=1$. On sites $j>i$, turn its nonidentity components into $Z_j$ and cancel them with $\mathrm{CNOT}_{j\to i}$, which preserves $X_i$. If $b_i$ still has $x_i=1$, apply $H_i$, then $S_i$, then $H_i$ to map it to $Z_i$ while preserving $X_i$. Previously processed pairs remain fixed because all later pairs are symplectically orthogonal to them. The accumulated gate list maps $F$ to the standard frame and implements $G_{\rm res}$.

Passing the algebraic checks does not by itself certify exact recovery: too few samples can produce additional equations that form a consistent but incorrect cut algebra. The uniform bound in Eq.~\eqref{sm:exact-kernel-prob} supplies the statistical guarantee. Duplicate samples may be discarded before elimination, since they do not change either kernel.

There are at most $\nu$ nonzero joint eigenspaces because each has even
dimension. The algorithm uses finite-field linear algebra and
Clifford-tableau operations throughout. It requires neither a search
over cuts or Clifford circuits nor recursive reconstruction of new
kernels. Straightforward elimination, for example, uses
$O(Mn^4+n^6)$ binary operations. The classical processing is therefore
polynomial in $n$ and $M$.

For the sample complexity, apply Eq.~\eqref{sm:exact-kernel-prob} with
$j=2$ to the original $n$-qubit samples. Exact recovery of this kernel
also gives the exact affine kernel. Lemma~\ref{sm:inheritance} then
guarantees that the residual quadratic kernel is exact, and
Theorem~\ref{sm:boolean} identifies the irreducible blocks and the Clifford-product
width. This proves the structure-recovery claim of
Theorem~\ref{thm:exact}, using
\begin{equation}
 M=O\!\left(\frac{n^2+\log(1/\delta)}{\gamma_2(\psi)}\right).
 \label{sm:exact-samples}
\end{equation}

\section{Tomography after frame recovery}
\label{sm:tomography-validation}

An exact decomposition gives pure block states, allowing us to improve the
tomography cost by using efficient pure-state estimators for every individual block. 
The gate-efficient pure-state estimator of Ref.~\cite[Appendix~C, Theorem~C.1]{HaahKothariODonnellTang2023} learns a $d$-dimensional pure state to infidelity $\epsilon_{\mathrm{inf}}$ using $O(d/\epsilon_{\mathrm{inf}})$ copies at constant success probability. Take the single-run success probability above $2/3$ and repeat $O(\log(1/\zeta))$ times. Select a candidate with more than half of the estimates within twice the single-run target trace-distance radius. With probability at least $1-\zeta$, a majority of the candidates are accurate; such a candidate exists and is within three target radii of the true state. Counting nearby estimates therefore amplifies confidence with only a constant-factor change in the target accuracy. Thus
$O((d/\epsilon_{\mathrm{inf}})\log(1/\zeta))$ copies suffice. Distances between candidate pure states can be evaluated from their classical descriptions, so this amplification also has polynomial classical cost.
To distribute the error across blocks, assign block $\alpha$
the infidelity budget
$\epsilon_{\mathrm{inf},\alpha}
=\varepsilon^2d_\alpha/\sum_\beta d_\beta$
and failure probability $\delta/r$. 
The global fidelity is the product
of the local fidelities, so these budgets ensure
\begin{equation}
 1-\prod_\alpha F_\alpha\le\sum_\alpha(1-F_\alpha)
 \le\varepsilon^2.
\end{equation}
We run all block tomography procedures in parallel on the decoded global copies.
The dimension-weighted error budgets give the same sufficient copy count
for every block, yielding
\begin{equation}
 m_{\rm tomo}^{\rm exact}
 =O\!\left(\frac{\sum_\alpha d_\alpha}{\varepsilon^2}
                  \log\frac r\delta\right)
 \le O\!\left(r2^k\varepsilon^{-2}\log\frac r\delta\right).
 \label{sm:exact-tomography-cost}
\end{equation}

\section{Approximate learning without a quadratic-gap assumption}
\label{sm:approx}

Exact reconstruction requires observing the rare Bell samples that violate
false quadratic equations. For approximate learning, these samples need not
be resolved if the additional empirical equations describe only weakly
entangled cuts. Here, we prove that this is the case for states with small
Clifford-product width and a nonzero affine gap, establishing
Theorem~\ref{thm:approx}. The affine gap separates the residual state from
acquiring another Pauli stabilizer. Throughout this section, we use normalized
trace distance, $\Dtr(\rho,\sigma)=\tfrac12\|\rho-\sigma\|_1$; its square
for two pure states is one minus their squared overlap.

\subsection{Accuracy of the empirical equations and stability of the projectors}

We first analyze the $\nu>0$ residual qubits after stabilizer removal, and denote their state by $|\psi\rangle\in\CTP_{\nu,k}$. This state has no nonidentity
Pauli stabilizer. To quantify its separation from acquiring a stabilizer,
write $R(v)=|\langle\psi|W_v|\psi\rangle|$ and assume
\begin{equation}
 \gamma_\star\le\gamma_1(\psi)=\frac{1-\max_{v\ne0}R(v)^2}{2},
 \qquad 0<\gamma_\star\le\tfrac12.
 \label{sm:approx-affine-margin}
\end{equation}
We draw $M$ independent Bell samples and compute the nullspace $K_M
\subseteq\RM(2,2\nu)$ of their quadratic feature matrix. This retains every
exact support equation, but may also retain equations violated by labels
that were not observed. The following bound controls all such equations
simultaneously. As in the main text, let $\tau>0$ be a uniform tolerance and write $\tau_f:=\Pr_{\xi\sim p_\psi}[f(\xi)=1]$ for an individual violation probability, suppressing the state argument. We write $D_2(\nu)=1+2\nu+\binom{2\nu}{2}$ for the number of
quadratic features.

\begin{lemma}[Uniform bound on the violation probability]
\label{sm:approx-soundness}
If
\begin{equation}
 M\ge\frac{D_2(\nu)\ln2+\ln(1/\delta)}{\tau},
 \label{sm:approx-sample-condition}
\end{equation}
then, with probability at least $1-\delta$, every $f\in K_M$ has violation
probability $\tau_f\le\tau$. For every sample realization,
$K_M$ also contains all exact quadratic equations.
\end{lemma}
\begin{proof}
A fixed equation with $\tau_f>\tau$ holds on all $M$ samples with
probability at most $e^{-M\tau}$. Taking a union bound over the
$2^{D_2(\nu)}$ quadratic Boolean functions gives the stated probability.
An exact equation holds on every label with nonzero probability and therefore always
belongs to $K_M$.
\end{proof}

Small violation probability alone does not identify a product cut. We next
show that the affine gap and the small-block promise force the polarization
of every sufficiently accurate equation to be an exact projector.

\begin{lemma}[Projectors from approximate equations]
\label{sm:approx-idempotence}
Under Eq.~\eqref{sm:approx-affine-margin}, a quadratic $f$ with
\begin{equation}
 \tau_f<\frac{\gamma_\star^3}{8\,2^k}
 \label{sm:approx-idempotence-threshold}
\end{equation}
has idempotent polarization: $E_f^2=E_f$.
\end{lemma}
\begin{proof}
We first relate the violation probability to Pauli correlations. Set
$E=E_f$, $\Delta=E^2+E$, and $|\chi\rangle=|\psi\rangle^{\otimes2}$.
The identity $\|(Z_f-I)|\chi\rangle\|_2=2\sqrt{\tau_f}$ implies that
conjugation by $Z_f$ changes the expectation of any observable with unit operator norm
by at most $4\sqrt{\tau_f}$. For the observable $W_v\otimes I$, which
translates Bell labels by $v$, we have
\begin{equation}
 Z_f(W_v\otimes I)Z_f
 \ \propto\ W_{(I+E)v}\otimes W_{Ev}^{*},
\end{equation}
where $f(y+v)+f(y)$ has linear part $[Ev,y]$, and an affine Bell character
with label $b$ is proportional to $W_b\otimes W_b^*$. Taking absolute values of the
expectations gives the approximate correlation identity
\begin{equation}
 \big|R(v)-R(Ev)R((I+E)v)\big|
 \le4\sqrt{\tau_f}.
 \label{sm:approx-conjugation}
\end{equation}

Suppose that idempotence fails, so $\Delta\ne0$. Choose a label $v$ with
$\Delta v\ne0$ that maximizes $R(v)$, and denote this maximum by $R_*$.
Both $Ev$ and $(I+E)v$ are nonzero, since otherwise $\Delta v=0$.
Their sum is $v$, so at least one also has nonzero $\Delta$-image and
hence expectation at most $R_*$. The affine gap bounds the expectation
of the other by $\sqrt{1-2\gamma_\star}$. Equation~\eqref{sm:approx-conjugation}
therefore gives
\begin{equation}
 R_*\le R_*\sqrt{1-2\gamma_\star}+4\sqrt{\tau_f},
 \qquad R_*\le\frac{4\sqrt{\tau_f}}{\gamma_\star},
 \label{sm:approx-max-descent}
\end{equation}
where $1-\sqrt{1-2\gamma_\star}\ge \gamma_\star$.

The small-block promise supplies a lower bound on the same correlation.
Consider any exact width-$k$ decomposition $V=\bigoplus_i V_i$ of the
state. We use this decomposition only in the proof; the learner does not
need to know it. Since $\Delta\ne0$, its restriction to some $V_i$ is
nonzero. Choose a linear functional on the image of that restriction
with nonzero pullback to $V_i$. By nondegeneracy of the local symplectic
form, the pullback is $[v,b]$ for some $0\ne b\in V_i$. Thus every
$v\in V_i$ with $[v,b]=1$ satisfies $\Delta v\ne0$.
Writing $d_i=2^{|A_i|}\le2^k$ and $R_i(v)=|\Tr(\rho_iW_v)|$, purity of
the exact block gives
\begin{equation}
 \sum_{[v,b]=1}R_i(v)^2
 =\frac{d_i}{2}\bigl[1-R_i(b)^2\bigr]\ge d_i \gamma_\star.
 \label{sm:approx-pauli-energy}
\end{equation}
To derive this identity, expand in Paulis to obtain
$\Tr(\rho_iW_b\rho_iW_b)=d_i^{-1}\sum_v(-1)^{[v,b]}R_i(v)^2$;
the left side equals $R_i(b)^2$ for a pure block, while
$d_i^{-1}\sum_vR_i(v)^2=1$. Subtracting these identities gives
Eq.~\eqref{sm:approx-pauli-energy}. The affine gap applies to $b$
because a Pauli supported on one exact block has the same expectation
in the full state. Since there are $d_i^2/2$ summands, at least one
violating label has a correlation satisfying
\begin{equation}
 R_*\ge\sqrt{\frac{2\gamma_\star}{d_i}}\ge\sqrt{\frac{2\gamma_\star}{2^k}}.
\end{equation}
This lower bound and Eq.~\eqref{sm:approx-max-descent} contradict
Eq.~\eqref{sm:approx-idempotence-threshold}. Hence $\Delta=0$ and
$E_f^2=E_f$.
\end{proof}

We can now recover compatible subsystems directly from the empirical
equations. Condition on the event in Lemma~\ref{sm:approx-soundness} and
choose $\tau<\gamma_\star^3/(8\,2^k)$. Every empirical polarization is idempotent.
Since the polarization space is linear, idempotence of $E+F$ also gives
$EF=FE$. Moreover, $\tau<\gamma_\star$ excludes every nonzero affine equation from
$K_M$, so polarization is injective on this space.

To obtain the subsystems, successively split $V$ into the zero and one
eigenspaces of a basis of these commuting maps. 
We use the eigenspace splitting and Clifford synthesis of
Algorithm~S1, but omit the check $r=\ell$ in step~5.
Unlike the exact cut algebra, the empirical polarization
space need not be closed under multiplication, so the
number of joint eigenspaces can exceed its dimension.
This gives a joint
decomposition $V=\bigoplus_{\alpha=1}^{r}V_\alpha$. Symplectic
self-adjointness makes distinct joint eigenspaces orthogonal and each
nonzero $V_\alpha$ nondegenerate, implying $r\le \nu$. Every true block
projector belongs to the empirical polarization space because its
partial-swap equation is exact. It therefore acts as a scalar on each
joint eigenspace. The true projectors are orthogonal and sum to $I$, so
each $V_\alpha$ lies within one true block. Thus the recovered blocks
have at most $k$ qubits, even when the empirical kernel contains
additional equations. Computing symplectic bases of these subspaces
gives a Clifford decoder by polynomial-time binary linear algebra.

\subsection{From compatible approximate cuts to a product approximation}

The recovered subsystems can split a true block along a weakly entangled
cut, so their reduced states need not be pure. To control the global
approximation error, we first bound the entropies across the cuts defined
by the empirical equations and then combine those bounds for the resulting blocks.

Write $d_K=\dim K_M$ and choose a basis $f_1,\ldots,f_{d_K}$ of $K_M$. In the decoded coordinates,
$E_{f_j}$ projects onto a union $A_j$ of recovered blocks, and $f_j$
differs from $q_{A_j}$ by an affine function. We can therefore apply the
swap inequality~\eqref{sm:twisted-bound} to obtain
\begin{equation}
 1-2\tau\le\langle Z_{f_j}\rangle
 \le|\langle Z_{f_j}\rangle|\le\Tr\rho_{A_j}^2.
 \label{sm:approx-cut-purity}
\end{equation}
For $\tau\le1/4$, this bounds the largest eigenvalue of $\rho_{A_j}$
from below by $1-2\tau$ and hence bounds its entropy in bits by
\begin{equation}
 S(A_j)\le h_2(2\tau)+2\nu\tau=:b_\nu(\tau).
 \label{sm:approx-cut-entropy}
\end{equation}
Here, $h_2$ is the binary entropy. Purity of the global state gives the
same bound for $A_j^c$. The polarization space acts diagonally on the
$r$ joint blocks and polarization is injective, so $d_K\le r\le \nu$.

Each recovered block is the intersection of at most $d_K$ sets chosen from
$A_j,A_j^c$. Strong subadditivity, applied to the disjoint subsystems
$A\setminus B$, $A\cap B$, and $B\setminus A$, gives
\begin{equation}
 S(A)+S(B)\ge S(A\cap B)+S(A\cup B).
\end{equation}
Since $S(A\cup B)\ge0$, we obtain $S(A\cap B)\le S(A)+S(B)$ for any two unions of the blocks. Applying this inequality successively gives
\begin{equation}
 \sum_{\alpha=1}^{r}S(\rho_\alpha)
 \le r d_K\,b_\nu(\tau)\le \nu^2b_\nu(\tau)=:B.
 \label{sm:approx-total-entropy}
\end{equation}
Thus low entropy across the empirical basis cuts controls the total local
entropy of the recovered blocks. The projectors onto the joint eigenspaces need not
themselves belong to the empirical polarization space.

To construct a nearby pure product, let $|\phi_\alpha\rangle$ be an eigenvector of $\rho_\alpha$ with largest eigenvalue $\lambda_\alpha$. The
corresponding local rank-one projectors commute and obey
$I-\prod_\alpha\Pi_\alpha\le\sum_\alpha(I-\Pi_\alpha)$. Combining this
inequality with $1-\lambda_\alpha\le S(\rho_\alpha)$ yields
\begin{equation}
 1-\left|\left\langle D^\dagger\psi\middle|
       \bigotimes_\alpha\phi_\alpha\right\rangle\right|^2
 \le\sum_\alpha(1-\lambda_\alpha)\le B.
 \label{sm:approx-product-overlap}
\end{equation}

We choose a tolerance that controls both idempotence and approximation
error:
\begin{equation}
 \tau=\min\left\{\frac{\gamma_\star^3}{32\,2^k},
                  \frac{\varepsilon^4}{1024\nu^4}\right\}.
 \label{sm:approx-tau}
\end{equation}
For $0<\varepsilon<1$, the bound $h_2(x)\le2\sqrt{x}$ gives
$B\le\varepsilon^2(\sqrt2/16+1/512)<\varepsilon^2/4$.
The recovered frame therefore contains a pure product within trace
distance $\varepsilon/2$ of the input. By
Lemma~\ref{sm:approx-soundness}, a sufficient measurement cost is
\begin{equation}
 M=O\!\left[(\nu^2+\log\delta^{-1})
     \max\left\{\frac{2^k}{\gamma_\star^3},\frac{\nu^4}{\varepsilon^4}\right\}\right]
 \label{sm:approx-structural-cost}
\end{equation}
Bell samples. Computing the empirical kernel and splitting the commuting
projectors both take time polynomial in $M$ and $\nu$. Thus this stage has
no dependence on the quadratic gap.

For an original input on $n$ qubits, recovering its $n-\nu$ exact Pauli stabilizers first requires $O((n+\log\delta^{-1})/\gamma_\star)$ Bell samples. The $\nu$-qubit residual inherits the affine gap and the width bound. We use independent batches for affine and quadratic recovery and divide the failure budget between the stages. Since $\nu\le n$, replacing $\nu$ by $n$ in Eq.~\eqref{sm:approx-structural-cost} bounds the total structural cost, including this preprocessing. Measuring the fixed qubits of a fresh decoded copy determines the stabilizer signs. If $\nu=0$, this completes the learning problem.

\subsection{Tomography of the approximate block states}
\label{sm:approx-tomography}

The frame recovery proves that a nearby pure product exists. To learn
that product, we must estimate its block states from the generally mixed
decoded marginals. Pure-state tomography guarantees do not apply directly
to these inputs. Instead, we estimate each marginal and extract an eigenvector with largest eigenvalue, which gives an explicit polynomial bound.

Write $d_\alpha=2^{|A_\alpha|}$, $d_{\max}\le2^k$, and
$\eta=\varepsilon^2/(8\nu)$. We estimate every local nonidentity Pauli mean
to additive accuracy $\eta/\sqrt{d_\alpha}$. The resulting Hermitian
Pauli-expansion estimate $\widetilde\rho_\alpha$ satisfies
\begin{equation}
 \|\widetilde\rho_\alpha-\rho_\alpha\|_\infty
 \le\|\widetilde\rho_\alpha-\rho_\alpha\|_2\le\eta,
 \label{sm:approx-local-estimate}
\end{equation}
where $\|\cdot\|_2$ denotes the Hilbert--Schmidt norm. Orthogonality of
the $d_\alpha^2$ Pauli matrices gives the last inequality. The raw
estimate need not be positive for this bound to hold.

We choose an eigenvector $|\widehat\phi_\alpha\rangle$ of the estimate with largest eigenvalue
as the learned pure block state. By the variational principle and
Eq.~\eqref{sm:approx-local-estimate}, its overlap with the true marginal
obeys
\begin{equation}
 \langle\widehat\phi_\alpha|\rho_\alpha|
           \widehat\phi_\alpha\rangle\ge\lambda_\alpha-2\eta.
\end{equation}
Applying the commuting-projector bound from
Eq.~\eqref{sm:approx-product-overlap} to these learned block states gives
infidelity at most $B+2r\eta<\varepsilon^2/2$. Thus the computed state
$D\bigotimes_\alpha|\widehat\phi_\alpha\rangle$ belongs to
$\CTP_{\nu,k}$ and approximates the input to trace distance at most
$\varepsilon$.

To determine the copy cost, apply Hoeffding's  inequality and a union
bound to all local Pauli means. This gives
\begin{equation}
 m_{\rm tomo}^{\rm approx}
 =O\!\left(\nu^2 2^{3k}\varepsilon^{-4}
                 \log\frac{\nu2^k}{\delta}\right).
 \label{sm:approx-tomography-cost}
\end{equation}
Each Pauli mean requires $O(d_\alpha\eta^{-2}
\log(\nu  d_{\max}^2/\delta))$ repetitions. We schedule at most
$d_{\max}^2$ settings, measuring across all blocks in parallel. One decoded
global copy supplies a copy of every marginal. The union bound does not
require independent outcomes across different blocks of that copy.
Constructing the measurement schedule and diagonalizing the estimates
take time polynomial in $\nu,2^k$ and the requested precision. Dividing
the failure budget among affine recovery, the empirical kernel, and
tomography completes the proof of Theorem~\ref{thm:approx}. Reattaching the $n-\nu$ fixed qubits and undoing the stabilizer decoder gives an $n$-qubit hypothesis in $\CTP_{n,k}$ with the same trace-distance error; its total cost is bounded by replacing $\nu$ with $n$ in the residual costs.

We have therefore obtained an approximate learner whose cost is
independent of $\gamma_2$, while retaining the affine-gap assumption.
The original input is exactly pure and belongs to $\CTP_{n,k}$. Removing the
dependence on $\gamma_1$, or allowing arbitrary noisy inputs near this
class, requires a further stability argument.

The argument for approximate equations also gives a direct interpretation of a
small quadratic gap: under the same promises, it certifies proximity to
an additional product cut.

\begin{corollary}[Small quadratic gap certifies an approximate cut]
\label{sm:approx-small-gap-cut}
Under Eq.~\eqref{sm:approx-affine-margin}, a state with
$\gamma_2(\psi)<\gamma_\star^3/(8\,2^k)$ has a nontrivial Clifford-frame bipartition
with reduced purity at least $1-2\gamma_2(\psi)$. Across that bipartition,
it is within trace distance $\sqrt{2\gamma_2(\psi)}$ of a product state.
\end{corollary}
\begin{proof}
Choose a quadratic with violation probability $\gamma_2>0$ and apply
Lemma~\ref{sm:approx-idempotence}. Its polarization is an idempotent
different from $0$ and $I$. Indeed, these two cases would make the
quadratic affine or affine plus the full-swap equation, giving
nonzero violation probability at least $\gamma_\star$. The swap bound then
gives a nontrivial cut with the stated purity. Its largest squared Schmidt
coefficient is at least that purity, which bounds the trace distance to
the corresponding product state.
\end{proof}

\section{Explicit examples and process learning}
\label{sm:applications}

\subsection{Exact gap for products of qubits}
\label{sm:qubitgaps}

We first determine the quadratic gap for Clifford-scrambled products of qubits. This gives explicit families for which the cost of exact learning can be evaluated. A pure qubit with Bloch vector $(r_x,r_y,r_z)$ has Bell probabilities
\begin{equation}
 p_I=\frac{1-r_y^2}{2},\qquad p_X=\frac{1-r_z^2}{2},\qquad
 p_Z=\frac{1-r_x^2}{2},\qquad p_Y=0.
 \label{sm:qubitbell}
\end{equation}
For a nonstabilizer qubit, all three allowed Bell labels have nonzero probability. Restricted to these labels, $1,x,z$ form a basis of the Boolean functions. We can therefore select any one label with an affine indicator.

\begin{proposition}[Qubit-product gap]
\label{sm:qubitgap}
Let $\psi=C\bigotimes_{i=1}^n\phi_i$, where $C\in\Cliff_n$ and each $\phi_i$ is a nonstabilizer qubit. Denote the smallest nonzero Bell probability of $\phi_i$ by $\eta_i$, and order these probabilities as $\eta_{(1)}\le\eta_{(2)}\le\cdots$. The quadratic gap is
\begin{equation}
 \gamma_2(\psi)=
 \begin{cases}\eta_1,&n=1,\\ \eta_{(1)}\eta_{(2)},&n\ge2.\end{cases}
 \label{sm:qubitgapformula}
\end{equation}
\end{proposition}
\begin{proof}
The lower bound follows from a property of independent variables with finite support. For each variable, consider the smallest positive probability of any of its values. A nonzero function of degree at most $\ell$ in these variables is nonzero with probability at least the product of the $\ell$ smallest of these probabilities, using all variables when there are fewer than $\ell$. Here, degree counts the distinct variables with nonconstant basis elements in a product-basis monomial.

To prove this property, choose a variable in a monomial of highest degree and the nonconstant basis element appearing on that variable in the chosen monomial. Its coefficient is a nonzero function of degree at most $\ell-1$ on the other variables. By induction, this coefficient is nonzero with at least the required product probability on the remaining variables. Whenever it is nonzero, the full function is nonzero on at least one possible value of the chosen variable. Multiplying these two probability bounds proves the claim.

A Boolean quadratic has degree at most two in these variables, which gives the lower bound in Eq.~\eqref{sm:qubitgapformula}. We attain the bound by multiplying affine indicators of the least-probable Bell labels on the two qubits with smallest $\eta_i$. For $n=1$, a single indicator suffices. Finally, the affine relabeling induced by a Clifford preserves both the polynomial degree and the violation probability.
\end{proof}

For $|T\rangle=(|0\rangle+e^{i\pi/4}|1\rangle)/\sqrt2$, the three nonzero probabilities are $1/2,1/4,1/4$. Thus $\gamma_2(C|T\rangle^{\otimes n})=1/16$ for every Clifford and $n\ge2$. Independent Haar-random qubits instead give
\begin{equation}
 \Pr[\eta_i\le y]=3(1-\sqrt{1-2y})=3y+O(y^2),\qquad 0\le y\le\tfrac14.
 \label{sm:qubithaarcdf}
\end{equation}
To obtain this distribution, note that each Bloch coordinate is uniform on $[-1,1]$. For $y\le1/4$, the three events $|r_\alpha|\ge\sqrt{1-2y}$ are disjoint except at boundaries of measure zero. Standard order statistics then give
\begin{equation}
 9n^2\gamma_2\ \xrightarrow{\ d\ }\ E_1(E_1+E_2),
 \label{sm:qubithaarlimit}
\end{equation}
where $E_1,E_2$ are independent exponential variables with unit rate. At any fixed ensemble quantile, $O(n^4)$ copies therefore suffice for exact structure recovery. The scaling is set by the two qubits closest to Pauli eigenstates. The cost is controlled by these rare local Bell samples.

\subsection{Learning random T-doped Clifford outputs}
\label{sm:random-t-learning}

We prove Corollary~\ref{cor:random-t-learning} by combining the known Clifford-product reduction for random doped circuits~\cite{FuxEtAl2025,LiuClark2026} with a constant Bell support gap. The structural reduction is prior work; the consequence here is efficient recovery from copies of the output without the preparation circuit. A circuit is sampled once, and all copies supplied to the learner are of that same pure state.

Let $T_1=\operatorname{diag}(1,e^{i\pi/4})$ on qubit one and define
\begin{equation}
 |\psi_j\rangle=T_1C_j|\psi_{j-1}\rangle,
 \qquad |\psi_0\rangle=|0\rangle^{\otimes n},
 \qquad C_j\overset{\rm iid}{\sim}\mathrm{Unif}(\Cliff_n).
 \label{sm:random-t-ensemble}
\end{equation}
An additional final Clifford does not affect the width or support gaps.

\paragraph{Probability of the product normal form.}
Consider the event $\mathcal G_j$ that each of the first $j$ rotations anticommutes with at least one remaining Pauli stabilizer. The disentangling construction of Refs.~\cite{FuxEtAl2025,LiuClark2026} then gives, inductively,
\begin{equation}
 |\psi_j\rangle=D_j\left(|T\rangle^{\otimes j}
 \otimes|0\rangle^{\otimes(n-j)}\right),
 \qquad D_j\in\Cliff_n,
 \label{sm:random-t-normal-form}
\end{equation}
up to global phase. To see the step explicitly, condition on $\mathcal G_j$ and the preceding circuit. The next $T$ gate acts in the decoded frame as $e^{-i\pi P/8}$, with $P=D_j^\dagger C_{j+1}^\dagger Z_1C_{j+1}D_j$. Anticommutation with a remaining stabilizer means that $P$ has $X$ or $Y$ on a qubit $i$ in $|0\rangle$. A local Clifford fixing $|0\rangle$ puts $P=X_i\otimes Q$, with $Q$ a Hermitian Pauli on the other qubits. The controlled-Pauli Clifford
\begin{equation}
 V_Q=|0\rangle\langle0|_i\otimes I
       +|1\rangle\langle1|_i\otimes Q
 \label{sm:random-t-controlled-pauli}
\end{equation}
satisfies, for any state $|\varphi\rangle$ of the other qubits,
\begin{equation}
 V_Qe^{-i\pi(X_i\otimes Q)/8}
       (|0\rangle_i\otimes|\varphi\rangle)
 =\left(\cos\frac\pi8|0\rangle-i\sin\frac\pi8|1\rangle\right)_i
       \otimes|\varphi\rangle.
 \label{sm:random-t-one-step}
\end{equation}
The new qubit is Clifford equivalent to $|T\rangle$, proving the induction.

Conditional on any such history, the unsigned label of $P$ is uniform over the $4^n-1$ nonidentity Paulis. The remaining $n-j$ independent stabilizers span an isotropic subspace of dimension $n-j$, whose symplectic orthogonal has dimension $n+j$. Exactly $2^{n+j}-1$ nonidentity labels commute with all of them. Thus
\begin{equation}
 \Pr(\mathcal G_{j+1}\mid\mathcal G_j)
 =1-p_j^{\rm fail},\qquad p_j^{\rm fail}=\frac{2^{n+j}-1}{4^n-1},
 \label{sm:random-t-step-probability}
\end{equation}
and, for $0\le t\le n$,
\begin{align}
 \Pr(\mathcal G_t)
 &=\prod_{j=0}^{t-1}\frac{4^n-2^{n+j}}{4^n-1}\notag\\
 &\ge1-\frac{2^n(2^t-1)-t}{4^n-1}
 \ge1-2^{t-n}.
 \label{sm:random-t-probability}
\end{align}
This sufficient event deliberately excludes rotations commuting with all remaining stabilizers, even when they preserve product structure. Its complement is therefore not a certificate of width greater than one. The bound tends to one when $n-t\to\infty$; in particular, $t\le n-c\log_2n$ gives exceptional probability at most $n^{-c}$. The condition $t<n$ alone does not give a probability tending to one through this bound. The calculation assumes independent uniform global Cliffords; it makes no corresponding assertion for arbitrary shallow local Clifford layers.

\paragraph{Bell gap and learning cost.}
On $\mathcal G_t$, Clifford invariance reduces the gap calculation to the product in Eq.~\eqref{sm:random-t-normal-form}. A $T$ qubit has nonzero Bell probabilities $1/2,1/4,1/4$, and a $|0\rangle$ qubit has two nonzero probabilities $1/2,1/2$. The finite-support argument in Proposition~\ref{sm:qubitgap} applies also to the latter two-point supports: affine functions span all functions on either local support, and restricting a Boolean quadratic gives a function of degree at most two in the independent single-qubit Bell labels. Every nonzero such function therefore has nonzero probability at least the product of the two smallest local minimum probabilities (or the single minimum if $n=1$). Products of affine indicators attain this bound. For $n\ge2$,
\begin{equation}
 \gamma_2(\psi_t)=
 \begin{cases}
  1/4,&t=0,\\
  1/8,&t=1,\\
  1/16,&2\le t\le n.
 \end{cases}
 \label{sm:random-t-gap}
\end{equation}
For $n=1$, the gaps are $1/2$ at $t=0$ and $1/4$ at $t=1$. In every case $\gamma_2\ge1/16$, so no additional support-gap promise is needed on $\mathcal G_t$.

Run the structural algorithm with a width-one cutoff: return failure if an algebraic consistency check fails or any recovered residual block has more than one qubit. This keeps the resource cost polynomial also on exceptional circuits. Apply Theorem~\ref{thm:exact} with failure budget $\delta/2$. It uses $O(n^2+\log\delta^{-1})$ Bell samples, each consuming two copies, to recover the stabilizer subspace and the $t$ residual single-qubit blocks. Together with the $n-t$ stabilized qubits, these give a Clifford decoder into $n$ single-qubit states.

Continuous-state tomography is unnecessary for this ensemble. By the stabilizer-removal argument in Sec.~\ref{sm:bell} and the uniqueness of the residual symplectic blocks in Theorem~\ref{thm:structure}, the decoded nonstabilizer qubits differ from the factors of $|T\rangle^{\otimes t}$ only by single-qubit Cliffords and a permutation. Their Bloch vectors are the twelve signed permutations of $(1/\sqrt2,1/\sqrt2,0)$. Any two distinct vectors differ in at least one coordinate by at least $1/\sqrt2$. Estimating each coordinate to error less than $1/(4\sqrt2)$ therefore identifies the corresponding vector uniquely. Measuring $X$, $Y$, and $Z$ on all decoded qubits in parallel and applying Hoeffding's inequality and a union bound achieves this accuracy with probability at least $1-\delta/2$, using $O(\log(n/\delta))$ fresh copies. Measuring $Z$ on the decoded stabilized qubits also determines their signs, which can be corrected by Pauli $X$ gates.

Absorb these sign corrections, the identified single-qubit Cliffords, and the permutation into the learned frame. The output is then a known Clifford $\widehat D$ satisfying
\begin{equation}
 |\psi_t\rangle=\widehat D\left(|T\rangle^{\otimes t}
 \otimes|0\rangle^{\otimes(n-t)}\right)
 \label{sm:random-t-exact-preparation}
\end{equation}
up to global phase. The procedure uses $O(n^2+\log\delta^{-1})$ copies in total and polynomial processing. Exact identification is possible because the local states lie in a finite, constantly separated set; it does not assert exact finite-copy tomography of arbitrary continuous states. Both stages succeed conditionally with probability at least $1-\delta$ for every circuit in $\mathcal G_t$. Averaging over the circuit gives
\begin{equation}
 \Pr[\text{successful learning}]
 \ge\Pr(\mathcal G_t)(1-\delta)
 \ge1-2^{t-n}-\delta,
 \label{sm:random-t-learning-success}
\end{equation}
which proves the corollary. Equation~\eqref{sm:random-t-exact-preparation} gives an exact preparation circuit for the output state. It does not identify the original random unitary or its action on other inputs.

\subsection{Exact learning of T-depth-one unitaries}
\label{sm:t-depth-one}

We prove Corollary~\ref{cor:t-depth-one}. The first step learns a Clifford preparation of the Choi state. The additional step is to recover a unitary circuit acting only on the output register. This conversion is essential: a Clifford preparation of a Choi state does not in general separate into input and output circuits.

We assume exact query access to the same unknown unitary $U$ on all $n$ qubits. The promise is that $U$ has the form in Eq.~\eqref{eq:t-depth-one}, with no initialized or discarded ancillas in that representation. The learner may use its own reference registers and known Clifford gates. Every query applies $U$ in the forward direction. All circuit equalities below are understood up to a global phase.

\paragraph{Step 1: recover an exact Clifford preparation of the Choi state.}
Set $N=2n$ and write
\begin{equation}
 |\chi_t\rangle=|T\rangle^{\otimes t}\otimes|0\rangle^{\otimes(N-t)}.
 \label{sm:t-seed}
\end{equation}
Moving the input Clifford to the reference register gives
\begin{equation}
 |J(U)\rangle=(C_{\rm in}^{T}\otimes C_{\rm out})
 \left[\bigotimes_{j=1}^t|J(T)\rangle_{R_jA_j}
 \otimes\bigotimes_{j=t+1}^n|J(I)\rangle_{R_jA_j}\right].
 \label{sm:t-choi-factorization}
\end{equation}
A CNOT from reference to output maps $|J(T)\rangle$ to $|T\rangle|0\rangle$ and $|J(I)\rangle$ to $|+\rangle|0\rangle$. Applying a Hadamard to each remaining $|+\rangle$ therefore exhibits an unknown Clifford $G_0$ such that $|J(U)\rangle=G_0|\chi_t\rangle$. In particular, $w_J(U)=1$ and $\nu(J(U))=t$.

The nonzero Bell probabilities of a $T$ qubit are $1/2,1/4,1/4$, while those of a fixed qubit are $1/2,1/2$. The finite-support argument of Proposition~\ref{sm:qubitgap}, including the fixed qubits as in Eq.~\eqref{sm:random-t-gap}, gives
\begin{equation}
 \gamma_2(J(U))\ge\frac1{16}.
 \label{sm:t-choi-gap}
\end{equation}
Theorem~\ref{thm:exact} thus recovers a Clifford decoder and all single-qubit blocks from $O(n^2+\log\delta^{-1})$ Bell samples. Each sample uses two forward queries. Stabilizer removal identifies $t$ without prior knowledge of the number of $T$ gates.

The uniqueness of the residual symplectic blocks in Theorem~\ref{thm:structure} implies that the decoded nonstabilizer qubits differ from the factors of $|T\rangle^{\otimes t}$ only by single-qubit Cliffords and a permutation. Their Bloch vectors are therefore the twelve signed permutations of $(1/\sqrt2,1/\sqrt2,0)$. Measuring $X,Y,Z$ on all decoded qubits in parallel to constant accuracy identifies these vectors exactly with probability at least $1-\delta/3$, using $O(\log(N/\delta))$ fresh Choi copies. The stabilized qubits can similarly be put in $|0\rangle$ with their measured signs. Absorbing the identified local Cliffords and permutation into the decoder yields a known Clifford $G$ with
\begin{equation}
 |J(U)\rangle=G|\chi_t\rangle.
 \label{sm:t-known-preparation}
\end{equation}
Here exact identification is possible because the local states belong to a finite, constantly separated set. It does not assert exact finite-copy tomography of arbitrary continuous states.

\paragraph{Step 2: move the learned rotation axes to the output register.}
The complete signed Pauli stabilizer group of the Choi state is generated by
\begin{equation}
 S_\ell=GZ_\ell G^\dagger,\quad t<\ell\le N.
 \label{sm:t-signed-stabilizers}
\end{equation}
The remaining operators $L_j=GZ_jG^\dagger$, $1\le j\le t$, specify the rotation axes of the learned magic qubits. We show that each can be multiplied by stabilizers to obtain a signed Pauli acting on the output register alone:
\begin{equation}
 L_j\prod_{\ell=t+1}^{N}S_\ell^{c_{j\ell}}
 =I_R\otimes Q_j,\qquad c_{j\ell}\in\F_2.
 \label{sm:t-output-representatives}
\end{equation}
The following argument establishes existence for the learned frame $G$, which need not equal the preparer's frame $G_0$.

Both $G_0$ and $G$ prepare the same state from $|\chi_t\rangle$. Thus $K=G_0^\dagger G$ preserves $|\chi_t\rangle$ and its complete stabilizer group $\langle Z_{t+1},\ldots,Z_N\rangle$. It induces a Clifford $K_{\rm log}$ on the $t$ logical qubits of the corresponding stabilizer code, and
\begin{equation}
 K_{\rm log}|T\rangle^{\otimes t}
 =e^{i\beta}|T\rangle^{\otimes t}.
 \label{sm:t-logical-symmetry}
\end{equation}
For $t>0$, the largest absolute expectation of a nonidentity Pauli on $|T\rangle^{\otimes t}$ is $1/\sqrt2$. The only labels attaining it are $X_j$ and $Y_j$, $1\le j\le t$: a Pauli with $w$ nonidentity factors from $\{X,Y\}$ has magnitude $2^{-w/2}$, and any $Z$ factor makes its expectation zero. Consequently, $K_{\rm log}$ permutes the set of these $2t$ labels. Within this set, each $X_j$ anticommutes only with $Y_j$. Preservation of commutation therefore permutes the pairs $\{X_j,Y_j\}$ and maps their sums, the $Z_j$ labels, to $Z_{\pi(j)}$. Including phases and fixed-qubit stabilizers gives
\begin{equation}
 KZ_jK^\dagger=\sigma_j Z_{\pi(j)}S_j^{(0)},
 \quad \sigma_j\in\{+1,-1\},\quad
 S_j^{(0)}\in\langle Z_{t+1},\ldots,Z_N\rangle.
 \label{sm:t-axis-permutation}
\end{equation}

For the explicit frame $G_0$ in Eq.~\eqref{sm:t-choi-factorization}, a magic-qubit $Z_j$ has a reference-register representative $C_{\rm in}^T Z_j C_{\rm in}^*$. The same logical operator has the output-register representative $C_{\rm out}Z_jC_{\rm out}^\dagger$, since their product is the stabilizer obtained from $Z_{R_j}Z_{A_j}$ of $|J(T)\rangle$. Equation~\eqref{sm:t-axis-permutation} shows that every axis of the learned frame belongs, up to sign, to one of these same stabilizer cosets. This proves Eq.~\eqref{sm:t-output-representatives}. For $t=0$, no axes need to be recovered.

The coefficients in Eq.~\eqref{sm:t-output-representatives} are computed without knowing $G_0$. Let $\lambda(P)\in\F_2^{2N}$ be the unsigned label of a Pauli $P$, and let $\pi_R$ restrict this label to the reference register. For each $j$, solve
\begin{equation}
 \sum_{\ell=t+1}^{N}c_{j\ell}\,\pi_R\lambda(S_\ell)
 =\pi_R\lambda(L_j)
 \label{sm:t-reference-cancellation}
\end{equation}
by binary elimination, then multiply the signed Paulis to obtain $Q_j$. The solution is unique: a nonidentity stabilizer with zero reference label would stabilize the maximally mixed output marginal of $|J(U)\rangle$, which is impossible. All $Q_j$ commute because the $L_j$ and $S_\ell$ commute. Their labels are independent because the $L_j$ are independent modulo the stabilizer group. Pauli phases must be retained in this step; they determine the signs of the rotation axes.

\paragraph{Step 3: reconstruct the remaining Clifford and the circuit.}
Define the known output-register unitary
\begin{equation}
 \mathcal R=\prod_{j=1}^t e^{-i\pi Q_j/8}.
 \label{sm:t-output-rotations}
\end{equation}
On the Choi state, $I_R\otimes Q_j$ acts as $L_j$, because their product is a stabilizer. These relations remain valid after any of the commuting rotations. Since $e^{+i\pi Z/8}|T\rangle$ is proportional to $|+\rangle$, it follows that
\begin{align}
 |J(\mathcal R^\dagger U)\rangle
 &=(I_R\otimes\mathcal R^\dagger)|J(U)\rangle\notag\\
 &\propto G\left(|+\rangle^{\otimes t}
       \otimes|0\rangle^{\otimes(N-t)}\right).
 \label{sm:t-clifford-choi}
\end{align}
The right-hand side is a known stabilizer state. The left-hand side is the Choi state of a unitary, so it is maximally entangled between reference and output. A maximally entangled stabilizer Choi state represents a Clifford unitary, denoted $C_*$.

For completeness, its circuit is obtained directly from its stabilizer tableau. The projection of the full $2n$-dimensional stabilizer-label space onto the reference labels is bijective: injectivity follows from the maximally mixed output marginal, and both spaces have dimension $2n$. For each reference generator $X_j$ or $Z_j$, solve for its unique signed stabilizer representative. The corresponding output Pauli is respectively $C_*X_jC_*^\dagger$ or $C_*Z_jC_*^\dagger$. These images specify the signed Clifford tableau and hence an efficiently synthesized Clifford circuit. This is a classical calculation; the learner need not physically apply $\mathcal R^\dagger$ to the unknown process.

Finally, extend the independent commuting labels of $Q_1,\ldots,Q_t$ to a symplectic basis and choose signs to synthesize a Clifford $A$ satisfying $AZ_jA^\dagger=Q_j$. Then
\begin{equation}
 U\propto\mathcal R C_*
 \propto A(T^{\otimes t}\otimes I)A^\dagger C_*.
 \label{sm:t-reconstructed-circuit}
\end{equation}
Returning $\widehat C_{\rm out}=A$ and $\widehat C_{\rm in}=A^\dagger C_*$ proves the circuit claim. The reconstructed Clifford layers need not equal the preparer's layers; their combined action with the recovered $T$ layer equals $U$ up to phase.

\paragraph{Resources and scope.}
Allocate the failure budget between exact structural recovery and the local Pauli measurements. Structural recovery uses $O(n^2+\log\delta^{-1})$ Bell samples and polynomial classical time. Local identification adds $O(\log(n/\delta))$ forward queries. Every later operation is binary linear algebra, signed Pauli multiplication, or Clifford synthesis on polynomial-size descriptions. Thus the total number of forward queries is $O(n^2+\log\delta^{-1})$, and all processing and known quantum operations have polynomial cost. No support-gap promise beyond the exact $T$-depth-one form is needed, because Eq.~\eqref{sm:t-choi-gap} is uniform over this class.

\subsection{Choi width and process predictions}
\label{sm:choi}

We next apply the state learner to unitary processes. A forward query on half of a maximally entangled state prepares the Choi state, allowing us to recover structure in an unknown process. For $d=2^n$, let $|\Phi_d\rangle=d^{-1/2}\sum_x|x\rangle_R|x\rangle_A$ and define
\begin{equation}
 |J(U)\rangle=(I_R\otimes U_A)|\Phi_d\rangle,\qquad
 \rho_J(U)=|J(U)\rangle\langle J(U)|,\qquad w_J(U)=\CPW(J(U)).
 \label{sm:choidefinition}
\end{equation}
One forward use of $U$ on register $A$, while retaining the reference $R$, gives one Choi copy. We therefore obtain one Bell sample with two queries. The protocol requires neither $U^\dagger$ nor controlled-$U$.

\begin{proposition}[Clifford-product-Clifford decomposition]
\label{sm:choireduction}
Let $U=C_{\rm out}(\bigotimes_\alpha U_\alpha)C_{\rm in}$, where the outer layers are Clifford and each $U_\alpha$ acts on at most $k$ qubits. Its Choi state satisfies
\begin{equation}
 |J(U)\rangle=(C_{\rm in}^{T}\otimes C_{\rm out})\bigotimes_\alpha|J(U_\alpha)\rangle,
 \qquad w_J(U)\le2k.
 \label{sm:choifactorization}
\end{equation}
Thus the minimum achievable largest block size $\kappa_{\rm layer}(U)$ of the product layer in such a decomposition obeys $\kappa_{\rm layer}(U)\ge\lceil w_J(U)/2\rceil$. Choi width is invariant under outer Clifford layers: $w_J(C_2UC_1)=w_J(U)$ for Clifford $C_1,C_2$.
\end{proposition}
\begin{proof}
Using $(I\otimes A)|\Phi_d\rangle=(A^T\otimes I)|\Phi_d\rangle$, we move the input Clifford onto the reference register. Its transpose is again Clifford. Each block Choi state has twice as many qubits as the corresponding unitary acts on, proving Eq.~\eqref{sm:choifactorization} and the bound on $\kappa_{\rm layer}$. The same identity shows that outer Clifford layers act as a Clifford on the Choi state and hence preserve its width.
\end{proof}

Applying Theorems~\ref{thm:exact} and~\ref{thm:approx} to the $2n$ Choi qubits gives two learning guarantees. An inverse-polynomial quadratic gap allows exact recovery of the Clifford frame in polynomial time, using $O(\gamma_2(J(U))^{-1}(n^2+\log\delta^{-1}))$ forward queries. With the affine-gap promise, we can instead learn an approximate description independently of the quadratic gap. For logarithmic Choi width, both methods produce a compact hypothesis $\widehat\rho_J$ satisfying $\Dtr(\widehat\rho_J,\rho_J(U))\le\varepsilon$, with polynomial resources under their respective gap assumptions.

For these general block promises, the output is an \emph{improper process hypothesis}~\cite{WadhwaEtAl2025}: the learned Clifford on the $2n$ Choi qubits need not factor into input and output layers. This possibility persists even when the Choi state is represented exactly; it concerns the recovered frame, not whether the target process is unitary. In contrast, failure to represent a unitary channel concerns the finite-accuracy estimate: although the positive Choi operator $d\widehat\rho_J$ defines a completely positive map $\widehat\Phi$, this map need not preserve trace or be unitary. At zero state error, $\widehat\rho_J=\rho_J(U)$ necessarily represents the target unitary. Exact structural recovery in Theorem~\ref{thm:exact} does not remove the statistical error of the subsequent block tomography. Our guarantee applies to the normalized Choi state and implies the normalized superoperator Frobenius bound
\begin{equation}
 d_F(\mathcal U,\widehat\Phi)
 :=\frac1{\sqrt2}\|\rho_J(U)-\widehat\rho_J\|_{\rm F}
 \le\sqrt2\,\varepsilon.
 \label{sm:choimetric}
\end{equation}
When the hypothesis is $\rho_J(V)$ for a unitary $V$, its trace distance from the target is $[1-|\Tr(U^\dagger V)|^2/d^2]^{1/2}$. The squared distance equals $(d+1)(1-F_{\rm avg}(U,V))/d$. These bounds do not give a dimension-independent guarantee in diamond norm.

Choi width captures a different restriction from bounded Clifford extent~\cite{DuttEtAl2026} or a small $T$-gate count~\cite{LaiCheng2022,leone2024learning}. It allows an extensive number of independent non-Clifford blocks. It does not, however, bound the accumulation of arbitrary interacting non-Clifford gates in a circuit.

The learned description allows us to efficiently predict any requested Pauli-transfer coefficient,
\begin{equation}
 R_{Q,P}(U)=d^{-1}\Tr(QUPU^\dagger)
 =\Tr[\rho_J(U)(P^T\otimes Q)].
 \label{sm:paulitransfer}
\end{equation}
To evaluate this expectation, we conjugate $P^T\otimes Q$ by the learned Clifford and factor the resulting Pauli over the small blocks. A Choi hypothesis with trace-distance error $\varepsilon$ predicts every coefficient with additive error at most $2\varepsilon$.

\section{Quadratic anticoncentration and pseudorandomness}
\label{sm:pseudorandomness}

We now efficiently distinguish states with an exact hidden product cut from Haar-random states using polynomially many copies, without learning the cut or assuming a gap. A hidden cut supplies an extra quadratic equation. For Haar-random states, every equation outside the span of the full swap has violation probability at least $1/4-o(1)$ with high probability.

To quantify this distinction, let $Q_n=\{v\in\F_2^{2n}:q(v)=0\}$, where $q(x,z)=x\cdot z$, and write $N_n=|Q_n|=2^{n-1}(2^n+1)$. The Bell vectors labeled by $Q_n$ form a basis of the symmetric two-copy subspace. Thus every pure-state Bell distribution is supported on $Q_n$.

\begin{lemma}[Quadratic weight on the symmetric quadric]
\label{sm:quadricweight}
For $n\ge2$, the smallest nonzero quadratic weight on $Q_n$ is
\begin{equation}
 d_Q(n):=\min_{f\in\RM(2,2n)\setminus\Span\{q\}}
 |\{v\in Q_n:f(v)=1\}|=2^{n-2}(2^{n-1}-1),\qquad
 \alpha_n:=\frac{d_Q(n)}{N_n}=\frac{2^{n-1}-1}{2(2^n+1)}.
 \label{sm:quadricdistance}
\end{equation}
\end{lemma}
\begin{proof}
We bound the quadratic weight using the Walsh sum $\mathcal W(g)=\sum_v(-1)^{g(v)}$. Expanding the indicator of $Q_n$ gives
\begin{equation}
 |\{v\in Q_n:f(v)=1\}|=\frac{N_n}{2}
 -\frac{\mathcal W(f)+\mathcal W(f+q)}4.
 \label{sm:restrictedweight}
\end{equation}
For a quadratic $g$ whose alternating polar form has rank $2r_g$, the symplectic normal form gives either $\mathcal W(g)=0$ or $|\mathcal W(g)|=2^{2n-r_g}$. Set $r=r_f$ and $r'=r_{f+q}$. Since the polar form of $q$ has full rank, rank subadditivity requires $r+r'\ge n$. When both $r,r'\ge1$, we obtain
\begin{equation}
 \mathcal W(f)+\mathcal W(f+q)
 \le2^{2n-r}+2^{2n-r'}\le2^{2n-1}+2^{n+1}.
\end{equation}
The last expression is maximized at $\{r,r'\}=\{1,n-1\}$. If one rank vanishes, the corresponding function is affine. A nonconstant affine function has zero Walsh sum, and adding $q$ gives Walsh magnitude $2^n$. The remaining allowed cases with a constant function are $f=1$ and $f=1+q$; both have negative Walsh sums and satisfy the same upper bound. Inserting this bound into Eq.~\eqref{sm:restrictedweight} proves the claimed minimum. It is attained by $f=x_1z_1$, whose Walsh sums are $2^{2n-1}$ and $2^{n+1}$.
\end{proof}

\begin{proposition}[Haar-state quadratic gap]
\label{sm:haarpure}
A Haar-random $n$-qubit pure state $\Psi$ satisfies
\begin{equation}
 |\gamma_2(\Psi)-\alpha_n|\le c_{\rm gap}n2^{-n/2},
 \label{sm:haarpuregap}
\end{equation}
with probability at least $1-e^{-c n^2}$. Moreover, $M=c_{\rm samp}D_2(n)$ Bell samples give $K_M=\Span\{q\}$ with probability at least $1-e^{-c n^2}$. Here $c,c_{\rm gap},c_{\rm samp}>0$ are universal constants and $n$ is sufficiently large.
\end{proposition}
\begin{proof}
For each quadratic $f$, let $\Pi_f$ project onto the Bell vectors in $Q_n$ for which $f(v)=1$. Its violation probability is $\tau_f(\psi)=\langle\psi|^{\otimes2}\Pi_f|\psi\rangle^{\otimes2}$. The Haar second moment equals the normalized projector onto the symmetric subspace, yielding
\begin{equation}
 \E\tau_f(\Psi)=\frac{|\{v\in Q_n:f(v)=1\}|}{N_n}.
\end{equation}
We next control fluctuations around this mean. Tensor squaring is $2$-Lipschitz on unit vectors, and the expectation of a projector is also $2$-Lipschitz. Thus $\tau_f$ is $4$-Lipschitz in Euclidean norm. Sphere concentration and a union bound over the $2^{D_2(n)}$ quadratics give
\begin{equation}
 \Pr\!\left[\max_f|\tau_f-\E\tau_f|>t\right]
 \le2^{D_2(n)+1}\exp(-c_0 2^nt^2).
 \label{sm:haarpureunion}
\end{equation}
Choosing $t=c_{\rm gap}n2^{-n/2}$ and applying Lemma~\ref{sm:quadricweight} bounds the violation probability of every quadratic outside $\Span\{q\}$ from below. The minimizing equation $f=x_1z_1$ gives the corresponding upper bound, proving Eq.~\eqref{sm:haarpuregap}.

For the empirical kernel, choose the constant $t=\alpha_n/2$ instead. Except on an event of probability $\exp[-\Omega(2^n)]$, every quadratic outside $\Span\{q\}$ is violated with probability at least $\alpha_n/2$. The probability that any such equation holds on all $M$ samples is therefore bounded by
\begin{equation}
 \Pr[K_M\supsetneq\Span\{q\}]
 \le\exp[-\Omega(2^n)]+2^{D_2(n)}(1-\alpha_n/2)^M.
 \label{sm:haarempirical}
\end{equation}
Taking $c_{\rm samp}$ sufficiently large makes this probability at most $e^{-c n^2}$.
\end{proof}

\subsection{A distinguisher for hidden product states}
\label{sm:prs}

To distinguish hidden product states from Haar-random states, we collect $M=c_{\rm samp}D_2(n)$ Bell samples, form the quadratic feature matrix, and accept when $\dim K_M>1$. A nontrivial product cut contributes a partial-swap equation independent of the full swap. Clifford relabeling preserves its quadratic degree, and this exact equation survives every set of samples. Proposition~\ref{sm:haarpure} shows that Haar acceptance is negligible.

The feature matrix has $M$ rows and $D_2(n)=O(n^2)$ columns. Constructing it costs $O(Mn^2)$ bit operations. Ordinary Gaussian elimination computes its rank in $O(MD_2(n)^2)$ bit operations; storing the matrix uses $O(MD_2(n))$ bits. With $M=O(n^2)$, these bounds become $O(n^6)$ time and $O(n^4)$ memory. 

This proves the exact-cut state claim in Corollary~\ref{cor:prs}: any ensemble with nonnegligible weight on $\CPW(\psi)<n$ is distinguishable from Haar with $O(n^2)$ copies and polynomial classical processing. The test requires no gap assumption. It extends the distinguishability result from physical product cuts~\cite{BoulandEtAl2025} to product cuts hidden by an arbitrary Clifford.

The test also detects sufficiently small nonzero gaps. A quadratic with violation probability $\gamma_2>0$ holds on all $M$ samples with probability $(1-\gamma_2)^M$, which lower-bounds acceptance. Hence an ensemble with nonnegligible weight on $\gamma_2=O(\log n/n^2)$ also fails pseudorandom-state security~\cite{JiLiuSong2018}. %

\subsection{Concentration for Haar-unitary Choi states}
\label{sm:haarunitary}

To obtain the unitary distinguisher, we need anticoncentration for Choi states of Haar-random unitaries. These states have maximally mixed reference marginals and are not Haar-random states on $2n$ qubits. We therefore calculate their Bell probabilities separately.

\begin{proposition}[Haar-unitary quadratic gap]
\label{sm:haarchoigap}
Let $U$ be Haar random in $\mathrm U(d)$, with $d=2^n$. Its Choi state obeys
\begin{equation}
 \frac14-c_{\rm gap}n2^{-n}\le\gamma_2(J(U))\le\frac14,
 \label{sm:haarchoigapformula}
\end{equation}
with probability at least $1-e^{-c n^2}$. Furthermore, $M=c_{\rm samp}D_2(2n)=O(n^2)$ Bell samples give $K_M=\Span\{q\}$ except with probability $e^{-c n^2}$. The constants $c,c_{\rm gap},c_{\rm samp}>0$ are universal.
\end{proposition}
\begin{proof}
We group the two Choi copies into the pair of reference registers and the pair of output registers, with Bell labels $a,b\in\F_2^{2n}$. Contracting the Bell vectors with the two maximally entangled states gives
\begin{equation}
 p_U(a,b)=\frac1{d^2}
 \left|\langle B_b|U^{\otimes2}|\overline{B_a}\rangle\right|^2.
 \label{sm:choibellprobability}
\end{equation}
Let $N_\pm=d(d\pm1)/2$ denote the dimensions of the symmetric and antisymmetric subspaces. These subspaces carry inequivalent irreducible representations of $U\mapsto U^{\otimes2}$. The two Bell vectors have swap eigenvalues $(-1)^{q(a)}$ and $(-1)^{q(b)}$, so Eq.~\eqref{sm:choibellprobability} vanishes for opposite parities. Within a common parity sector $\sigma\in\{+,-\}$, Schur twirling gives the exact mean probability
\begin{equation}
 \E_U p_U(a,b)=\frac1{d^2N_\sigma}
 \quad\text{if }(-1)^{q(a)}=(-1)^{q(b)}=\sigma.
 \label{sm:choibellmean}
\end{equation}
Complex conjugation preserves the swap parity of a Bell vector. Thus the allowed labels form $Q_{2n}=\{(a,b):q(a)+q(b)=0\}$, and each has mean probability at least $1/(d^2N_+)$. Applying Lemma~\ref{sm:quadricweight}, we find that every $f\notin\Span\{q\}$ satisfies
\begin{equation}
 \E_U\tau_f(J(U))\ge\frac{d_Q(2n)}{d^2N_+}
 =\frac{d^2-2}{4d(d+1)}=\frac14-O(d^{-1}).
 \label{sm:choimeanlower}
\end{equation}

To control deviations from this mean, use $\|J(U)-J(V)\|_2=d^{-1/2}\|U-V\|_{\rm HS}$. Together with the $4$-Lipschitz bound for $\tau_f$, this makes $U\mapsto\tau_f(J(U))$ $4/\sqrt d$-Lipschitz in Hilbert--Schmidt distance. Applying concentration on $\mathrm{SU}(d)$~\cite[Proposition 2.2]{MeckesMeckes2013} to this function and its negative gives
\begin{equation}
 \Pr[|\tau_f-\E\tau_f|>t]\le2\exp(-c_0d^2t^2).
 \label{sm:choiconcentration}
\end{equation}
The same bound holds for Haar $\mathrm U(d)$: multiplying a Haar special unitary by an independent uniform phase gives a Haar unitary, and $\tau_f$ is invariant under this phase. We now take $t=c_{\rm gap}n/d$ and a union bound over the $2^{D_2(2n)}$ quadratics. This proves the lower bound in Eq.~\eqref{sm:haarchoigapformula} and excludes additional exact equations with the stated probability.

For the upper bound, each reference qubit is maximally mixed for every unitary $U$. Its partial-swap equation therefore has nonzero violation $(1-\Tr[(I/2)^2])/2=1/4$. To obtain the empirical-kernel statement, choose a constant $t$ in Eq.~\eqref{sm:choiconcentration} and apply the union bound in Eq.~\eqref{sm:haarempirical} to $2n$ qubits.
\end{proof}

\begin{corollary}[Hidden Choi cuts exclude pseudorandom unitaries]
\label{sm:pru}
Any efficiently generated ensemble of $n$-qubit unitaries with nonnegligible probability on $w_J(U)<2n$ can be distinguished from a Haar-unitary oracle using $O(n^2)$ forward queries and polynomial processing. In particular, the conclusion applies to ensembles with a product layer between two Clifford circuits when every block in that layer acts on at most $k<n$ qubits.
\end{corollary}
\begin{proof}
We prepare Choi copies with forward queries and apply the quadratic nullspace test to their $2n$ qubits. If $w_J(U)<2n$, an extra exact quadratic survives every sample, so the test accepts with certainty. Proposition~\ref{sm:haarchoigap} bounds its Haar acceptance probability by $e^{-\Omega(n^2)}$. For the stated Clifford-product-Clifford decompositions, Eq.~\eqref{sm:choifactorization} gives $w_J(U)\le2k<2n$. The observer applies $U$ only to a chosen register entangled with a retained reference; inverse and controlled oracles are unnecessary.
\end{proof}

\section{Classical Bell sampling from MPSs and Clifford-augmented MPSs}
\label{sm:mps-bell}

If the input is supplied as an MPS, the Bell samples used by our algorithms can be generated classically. We give a sequential sampler with cost $O(n\chi^3)$ per sample for an open-boundary MPS of maximum bond dimension $\chi$. The construction uses canonical-form sampling~\cite{FerrisVidal2012} and is closely related to perfect Pauli sampling of MPSs~\cite{LamiCollura2023Sampling}. Here the required distribution is specifically
\begin{equation}
 p_\psi(v)=2^{-n}|\langle\psi^*|W_v|\psi\rangle|^2,
 \label{sm:mps-bell-target}
\end{equation}
corresponding to Bell measurements on $|\psi\rangle\otimes|\psi\rangle$. For a complex state, this generally differs from the Pauli distribution $2^{-n}|\langle\psi|W_v|\psi\rangle|^2$. The distinction determines which conjugations appear in the contraction below. We then extend the sampler to a supplied Clifford-augmented MPS by an affine relabeling of the samples.

\subsection{An MPS for the Bell amplitudes}

Write a normalized open-boundary MPS as
\begin{equation}
 |\psi\rangle=\sum_{a_1,\ldots,a_n\in\{0,1\}}
 A_1^{a_1}A_2^{a_2}\cdots A_n^{a_n}|a_1\cdots a_n\rangle,
 \qquad A_j^a\in\mathbb C^{\chi_{j-1}\times\chi_j},
 \quad\chi_0=\chi_n=1.
 \label{sm:mps-input}
\end{equation}
We choose a right-canonical gauge,
\begin{equation}
 \sum_{a=0}^1 A_j^a(A_j^a)^\dagger=I_{\chi_{j-1}}.
 \label{sm:mps-right-canonical}
\end{equation}
A sweep of QR or singular-value decompositions puts the MPS in this form in $O(n\chi^3)$ operations without truncating the state. The canonical condition means that the unmeasured suffix states indexed by the left virtual bond are orthonormal.

At site $j$, the two copies have the tensor $A_j^a\otimes A_j^b$. For a single-site Bell label $v_j=(x_j,z_j)$, let
$\beta_v^{ab}=\langle B_v|a,b\rangle$. Our convention
$|B_v\rangle=(I\otimes W_v)(|00\rangle+|11\rangle)/\sqrt2$ gives
$\beta_v^{ab}=(W_v)_{ba}^*/\sqrt2=(W_v)_{ab}/\sqrt2$, since $W_v$ is Hermitian. Thus the Bell-amplitude tensors are
\begin{equation}
 \mathcal A_j^v=\sum_{a,b=0}^1\beta_v^{ab}
                    A_j^a\otimes A_j^b,
 \qquad
 \langle B_{v_1}\cdots B_{v_n}|\psi\rangle^{\otimes2}
 =\mathcal A_1^{v_1}\cdots\mathcal A_n^{v_n}.
 \label{sm:mps-bell-tensor}
\end{equation}
The amplitude MPS has virtual dimension at most $\chi^2$. Its local physical dimension is four, one for each single-pair Bell label. Completeness of the Bell basis and Eq.~\eqref{sm:mps-right-canonical} imply
\begin{equation}
 \sum_v\mathcal A_j^v(\mathcal A_j^v)^\dagger
 =\sum_{a,b}A_j^a(A_j^a)^\dagger\otimes
                    A_j^b(A_j^b)^\dagger
 =I_{\chi_{j-1}^2}.
 \label{sm:mps-bell-canonical}
\end{equation}
Consequently, summing over all single-pair Bell labels to the right of a chosen prefix contracts the suffix to the identity. This is what makes its conditional probabilities inexpensive to evaluate.

\subsection{Sequential conditional probabilities}

After drawing $v_1,\ldots,v_{j-1}$, store the normalized prefix amplitude as a matrix
$L_{j-1}\in\mathbb C^{\chi_{j-1}\times\chi_{j-1}}$, with
$\|L_{j-1}\|_{\rm F}=1$. Its two indices are the virtual indices of the two copies. Initially $L_0=(1)$. For each candidate single-pair Bell label $v$, contract the next site to obtain
\begin{equation}
 T_j^v=\sum_{a,b=0}^1\beta_v^{ab}
             (A_j^a)^T L_{j-1}A_j^b,
 \qquad
 \pi_j(v\mid v_1,\ldots,v_{j-1})=\|T_j^v\|_{\rm F}^2.
 \label{sm:mps-conditional}
\end{equation}
The transpose in this expression is an ordinary transpose, not an adjoint: both copies contribute the original ket tensors. Complex conjugation enters when taking the squared Frobenius norm. Equation~\eqref{sm:mps-bell-canonical} gives $\sum_v\pi_j(v\mid v_{<j})=1$.

For implementation, define the four matrices
$K_j^{ab}=(A_j^a)^T L_{j-1}A_j^b$. The candidate contractions are
\begin{align}
 T_j^I&=(K_j^{00}+K_j^{11})/\sqrt2,
 &T_j^Z&=(K_j^{00}-K_j^{11})/\sqrt2,\nonumber\\
 T_j^X&=(K_j^{01}+K_j^{10})/\sqrt2,
 &T_j^Y&=i(K_j^{10}-K_j^{01})/\sqrt2.
 \label{sm:mps-four-outcomes}
\end{align}
The labels $I,X,Z,Y$ correspond to $(0,0),(1,0),(0,1),(1,1)$, respectively. Draw $v_j$ from these four probabilities and update
\begin{equation}
 L_j=\frac{T_j^{v_j}}{\sqrt{\pi_j(v_j\mid v_{<j})}}.
 \label{sm:mps-prefix-update}
\end{equation}
A zero-probability branch is never selected. Repeat until $j=n$, and return the binary sample $v=(x_1,\ldots,x_n,z_1,\ldots,z_n)$. Reset $L_0$ and use fresh random draws for each independent sample.

To verify the distribution, let $U_j$ be the unnormalized prefix obtained by contracting the first $j$ tensors in Eq.~\eqref{sm:mps-bell-tensor}, reshaped as a matrix. The right-canonical suffix makes the probability of that prefix exactly $\|U_j\|_{\rm F}^2$. Equation~\eqref{sm:mps-conditional} is therefore the ratio of the probabilities of the extended and preceding prefixes. Their product telescopes:
\begin{equation}
 \prod_{j=1}^n\pi_j(v_j\mid v_{<j})
 =\left|\mathcal A_1^{v_1}\cdots\mathcal A_n^{v_n}\right|^2
 =p_\psi(v).
 \label{sm:mps-sampling-correctness}
\end{equation}
Thus the sampler produces the same distribution as the physical two-copy Bell measurement, without rejection sampling or a Markov chain.

\subsection{Cost and extension to Clifford-augmented MPSs}

Each matrix $K_j^{ab}$ requires matrix multiplications of dimensions at most $\chi$. Thus evaluating the four conditional probabilities and updating the prefix costs $O(\chi^3)$ arithmetic operations per site. %
One Bell sample costs $O(n\chi^3)$ arithmetic operations and $O(\chi^2)$ working memory, in addition to the $O(n\chi^2)$ storage for the input MPS and $O(n)$ for the sample. Canonicalizing the input once gives a total cost of $O((M+1)n\chi^3)$ for $M$ independent samples.

Now suppose that the supplied representation is
\begin{equation}
 |\psi\rangle=C|\varphi\rangle,
 \label{sm:camps-input}
\end{equation}
where $C$ is a known Clifford and $|\varphi\rangle$ is a normalized MPS of maximum bond dimension $\chi_{\rm in}$. Equation~\eqref{sm:affine-action} gives
\begin{equation}
 p_\psi(F_Cv+t_C)=p_\varphi(v),
 \label{sm:camps-bell-distribution}
\end{equation}
where $t_C$ is the Pauli label of $CC^T$. We therefore sample $v$ from the residual MPS using the procedure above and compute $F_Cv+t_C$, which takes $O(n^2)$ time. Repeating these two steps with independent samples gives the Bell distribution of the full Clifford-augmented state, without applying $C$ to its MPS tensors.

\section{Clifford-product width}
\label{sm:nsee}

Our decomposition identifies independent subsystems after a global Clifford transformation. We relate Clifford-product width to the zero set of the summed nonstabilizerness entanglement entropy (NsEE)~\cite{HuangQianQin2025NsEE}. For a pure $n$-qubit state, define
\begin{equation}
 E_{\mathrm{Ns}}(\psi)
 =\min_{C\in\Cliff_n}\sum_{j=1}^{n-1}S\bigl(\rho_{[j]}^C\bigr),
 \qquad
 \rho_{[j]}^C=\Tr_{[j]^c}\!\left(C|\psi\rangle\langle\psi|C^\dagger\right),
 \label{sm:nsee-definition}
\end{equation}
where $[j]=\{1,\ldots,j\}$ and $S(\rho)=-\Tr(\rho\log_2\rho)$. The same Clifford is used for all cuts: the minimization is outside the sum. The optimization includes qubit permutations, so the quantity does not depend on the initial ordering. The sum is empty for $n=1$.

\begin{proposition}[Vanishing NsEE and width one]
\label{sm:nsee-zero}
For every pure $n$-qubit state,
\begin{equation}
 E_{\mathrm{Ns}}(\psi)=0
 \quad\Longleftrightarrow\quad
 \CPW(\psi)=1.
 \label{sm:nsee-zero-equivalence}
\end{equation}
\end{proposition}
\begin{proof}
If $\CPW(\psi)=1$, there is a Clifford $D$ such that $D^\dagger|\psi\rangle=\bigotimes_{i=1}^n|\phi_i\rangle$. Choosing $C=D^\dagger$ in Eq.~\eqref{sm:nsee-definition} makes every reduced state $\rho_{[j]}^C$ pure. All entropies vanish, and their nonnegativity gives $E_{\mathrm{Ns}}(\psi)=0$.

Conversely, the Clifford group is finite modulo global phases, so the minimum in Eq.~\eqref{sm:nsee-definition} is attained by some $C_*$. If this minimum is zero, each nonnegative summand is zero. Write $|\chi\rangle=C_*|\psi\rangle$. Purity of its first-qubit reduced state implies, by the Schmidt decomposition,
\begin{equation}
 |\chi\rangle=|\phi_1\rangle\otimes|\chi_{2\ldots n}\rangle.
 \label{sm:nsee-first-factor}
\end{equation}
Suppose inductively that the first $j-1$ qubits have been separated. The reduced state on $[j]$ is then the tensor product of these pure single-qubit states and the reduced state of qubit $j$ in the remaining pure state. Since its entropy is zero, that single-qubit reduced state is also pure. A further Schmidt decomposition separates qubit $j$. Iterating through $j=n-1$ gives $|\chi\rangle=\bigotimes_{i=1}^n|\phi_i\rangle$, and hence $\CPW(\psi)=1$. For $n=1$, both statements hold directly.
\end{proof}

\end{document}